\documentclass[lettersize,journal]{IEEEtran}
\usepackage{amsmath,amsfonts}
\usepackage{amsmath, amssymb}
\usepackage{algorithm}
\usepackage{algpseudocode}

\usepackage{array}
\usepackage{caption}
\usepackage{subcaption}
\usepackage{textcomp}
\usepackage{stfloats}
\usepackage{placeins}
\usepackage{url}
\usepackage{verbatim}
\usepackage{graphicx}
\usepackage{cite}
\usepackage{balance}
\usepackage{xcolor}
\usepackage{hyperref}
\usepackage{cleveref}
\usepackage{titlesec}
\titlespacing{\section}{0pt}{1\baselineskip}{0.5\baselineskip}
\titlespacing{\subsection}{0pt}{0.44 \baselineskip}{0.3\baselineskip}

\hypersetup{
	colorlinks=true,
	citecolor=blue,  
	linkcolor=blue,  
	urlcolor=blue    
}

\makeatletter
\patchcmd{\@maketitle}
{\addvspace{0.5\baselineskip}\egroup}
{\addvspace{-1\baselineskip}\egroup}
{}
{}
\makeatother

\begin{document}
	\setlength{\abovedisplayskip}{5pt}
	\setlength{\belowdisplayskip}{5pt}
	\setlength{\abovedisplayshortskip}{5pt}
	\setlength{\belowdisplayshortskip}{5pt}

	\title{Overcoming BS Down-Tilt for Air-Ground ISAC Coverage: Antenna Design, Beamforming and User Scheduling}
		
		\author{Lingyi Zhu,~\IEEEmembership{Student Member,~IEEE}, Zhongxiang Wei,~\IEEEmembership{Senior Member,~IEEE}, Fan Liu,~\IEEEmembership{Senior Member,~IEEE}, Jianjun Wu,~\IEEEmembership{Member,~IEEE}, Xiao-Wei Tang,~\IEEEmembership{Member,~IEEE}, Christos Masouros,~\IEEEmembership{Fellow,~IEEE}, \\and Shanpu Shen,~\IEEEmembership{Senior Member,~IEEE}

		\thanks{Lingyi Zhu, Zhongxiang Wei and Xiao-Wei Tang are with the College of Electronics and Information Engineering, Tongji University, Shanghai 201804, China (e-mail: {2431933@tongji.edu.cn}; {z\_wei@tongji.edu.cn}; {xwtang}@tongji.edu.cn). }
		
		\thanks{F. Liu is with the National Mobile Communications Research Laboratory, Southeast University, China (e-mail: {f.liu}@ieee.org).}
		
		\thanks{Jianjun Wu is with the Wireless	Communications Lab, Asmote technologies Co., Ltd., Shanghai, China (e-mail: {jjun}\_wu08@126.com).}
		
		\thanks{C. Masouros is with the Department of Electronic and Electrical Engineering, University College London, London WC1E 7JE, U.K. (e-mail: {c.masouros}@ucl.ac.uk).}
		
		\thanks{S. Shen is with the State Key Laboratory of Internet of Things for Smart City and Department of Electrical and Computer Engineering, University of Macau, Macau, China (email: {shanpushen}@um.edu.mo).}
		\vspace{-1pt} 
	}

	\maketitle
	
	\begin{abstract}
		Integrated sensing and communication holds great promise for low-altitude economy applications. However, conventional downtilted base stations primarily provide sectorized forward lobes for ground services, failing to sense air targets due to backward blind zones. In this paper, a novel antenna structure is proposed to enable air-ground beam steering, facilitating simultaneous full-space sensing and communication (S\&C). Specifically, instead of inserting a reflector behind the antenna array for backlobe mitigation, an omni-steering plate is introduced to collaborate with the active array for omnidirectional beamforming. Building on this hardware innovation, sum S\&C mutual information (MI) is maximized, jointly optimizing user scheduling, passive coefficients of the omni-steering plate, and beamforming of the active array. The problem is decomposed into two subproblems: one for optimizing passive coefficients via Riemannian gradient on the manifold, and the other for optimizing user scheduling and active array beamforming. Exploiting relationships among S\&C MI, data decoding MMSE, and parameter estimation MMSE, the original subproblem is equivalently transformed into a sum weighted MMSE problem, rigorously established via the Lagrangian and first-order optimality conditions. Simulations show that the proposed algorithm outperforms baselines in sum-MI and MSE, while providing $360^{\circ}$ sensing coverage. Beampattern analysis further demonstrates effective user scheduling and accurate target alignment.
	\end{abstract}


	\section{Introduction} \label{sec:intro}
	With the ever-growing demand for spectrum  efficiency, integrated sensing and communication (ISAC) has drawn significant attention for its ability to unify radar sensing and wireless communication systems \cite{ref1,ref2,ref3}. By leveraging advanced multiple-input and multiple-output (MIMO) architectures, ISAC is expected to support high-capacity, low-latency and reliable communications while simultaneously enabling ultra-accurate sensing, thus emerging as a key enabler for next-generation networks \cite{ref4,ref5,ref46}. Beyond performance improvements, ISAC also unlocks transformative opportunities across both civilian and military domains \cite{ref1,ref5,ref6}. One particularly promising application lies in the low-altitude economy \cite{ref7}, a comprehensive economic form driven by the extensive deployment of aircraft to support various services such as urban logistics, security surveillance, environmental observation and communication \cite{ref7,ref8,ref9}. The efficient operation of these airspace activities fundamentally relies on monitoring diverse air objects\cite{ref10,ref11}. Consequently, achieving ubiquitous and reliable low-altitude sensing has become an urgent and strategic requirement \cite{ref47}.

	In this context, ISAC-enabled base stations (BSs) are envisioned for dual roles: ensuring connectivity to ground users and providing sensing capabilities to detect and localize air targets \cite{ref48}. The antennas of conventional BS are designed for sectorized and forward coverage \cite{ref12,ref13,ref14}, as illustrated in Fig.~\ref{fig1}(a). In order to suppress backlobe beams and enhance forward beamforming gain, a backlobe reflector plate is typically placed behind the antenna array. Also, practically, BS antennas are often configured with a downtilt of around 10 degrees for covering ground users\cite{ref15,ref16}. A principal limitation is that sensing and communication (S\&C) coverage is inherently constrained to the front half-space. It creates a back-side blind zone, that hinders reliable tracking of targets along their full trajectories, as illustrated in Fig.~\ref{fig1}(b). Additionally, the backlobe can not be fully mitigated, where residual backlobe still causes power leakage and interference which harms sensing performance in the air \cite{ref17}. To overcome these issues, some works (e.g., \cite{ref18,ref19}) have introduced dedicated ISAC BSs for unmanned aerial vehicle (UAV) applications, featuring upward-facing antennas solely responsible for air S\&C. While such designs improve UAV-specific performance, they do not support ground communication and necessitate additional infrastructure, resulting in higher deployment costs, extra hardware and increased energy consumption. Moreover, these systems still operate within half-space service architectures, which are insufficient for comprehensive full-space monitoring. Particularly when deployed on urban rooftops, these BSs often fail to detect air targets below their field of view, creating persistent near-ground blind zones. As such, the above challenges highlight the urgent need for a more flexible antenna architecture capable of full-space beam steering.

	\begin{figure*}[!t]
		\centering
		\includegraphics[width=0.99\textwidth]{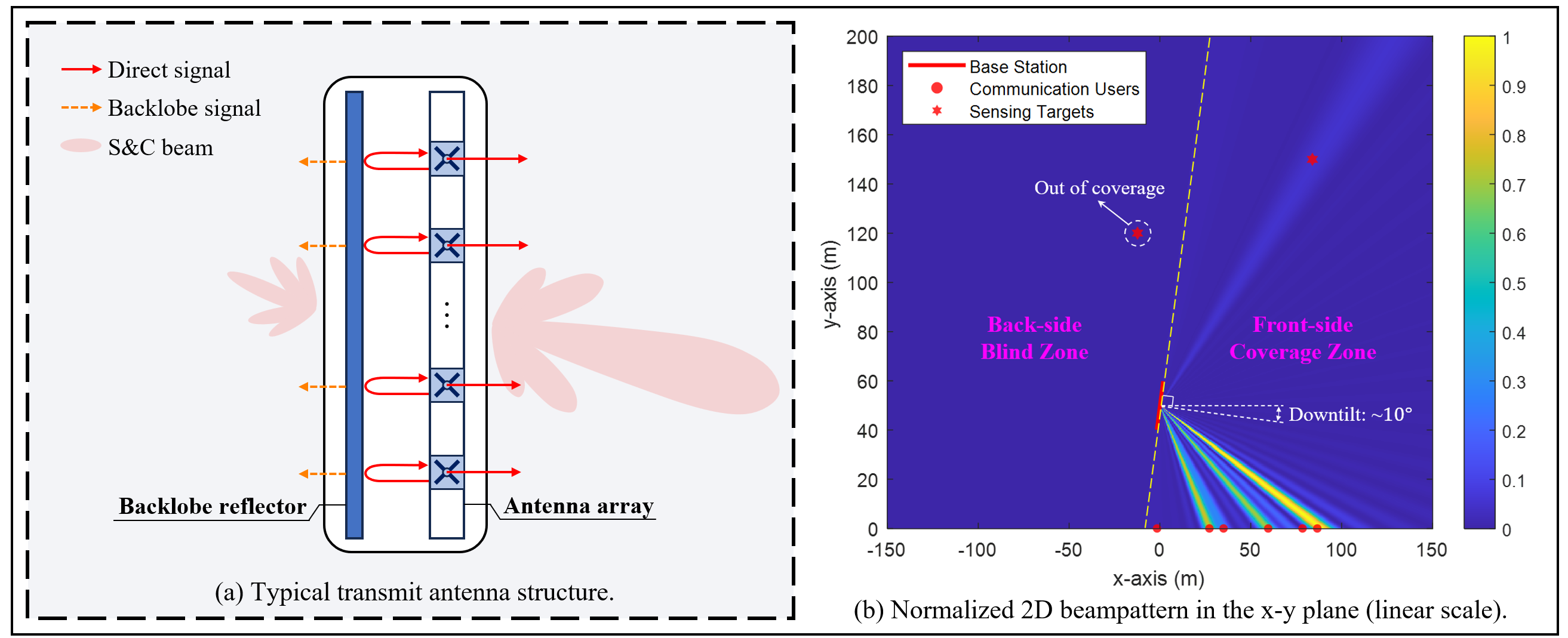} 
		\caption{
			Conventional ISAC base station antenna structure and its sensing coverage limitations.
		}
		\label{fig1}
		\vspace{-0.4cm} 
	\end{figure*}

	Apart from antenna architecture, the design of dual-functional beamforming is equally crucial for realizing efficient ISAC systems. Early work \cite{ref20} minimizes downlink multiuser interference while approximating a desired radar waveform for MU-MIMO scenarios, \cite{ref12} targeted an ideal radar beampattern subject to per-user signal-to-interference-plus-noise ratio (SINR) constraints for communications. These methods, however, rely on prior knowledge of sensing pattern and fail to directly quantify sensing performance. More recent works adopt task-dependent reliability metrics for sensing, e.g., detection probability and Cramér–Rao bound (CRB), and efficiency metrics for communication, e.g., SINR and achievable rate. \cite{ref21} proposed a design that maximizes the minimum detection probability across multiple targets under SINR constraints. \cite{ref13} minimized the angle estimation CRB under SINR constraints. Other efforts, e.g., \cite{ref22}, considered radar-specific metrics including transmit beampattern mean squared error (MSE) and signal-to-clutter-plus-noise ratio (SCNR), while maximizing spectral efficiency.

	In general, most of these works optimize either sensing or communication performance while treating the other as a constraint, often incorporating it into the objective via penalty terms and solving the resulting problem with efficient algorithms \cite{ref12,ref13,ref20,ref21,ref22}. Although this approach establishes a basic trade-off between the two functionalities, the choice of constraint thresholds and penalty weights inevitably locks the system into a fixed balance, limiting flexibility across different operating conditions. Consequently, Pareto-based optimization frameworks treat sensing and communication as a genuine multi-objective problem. The Pareto frontier captures all non-dominated operating points, providing a more systematic way to explore the trade-off between S\&C performance \cite{ref23,ref24}. Nevertheless, the heterogeneity of S\&C metrics still makes the resulting Pareto trade-offs difficult to interpret in a clear and intuitive manner. More importantly, in multi-objective optimization problems, if two objective functions exhibit inconsistencies in physical units or significant disparities in numerical magnitudes, it may result in ill-conditioning of the overall objective function (e.g., drastic gradient variations, steep function surfaces). These challenges, in turn, motivate the adoption of a unified information-theoretic metric, i.e., mutual information (MI). Specifically, communication MI corresponds to the achievable rate, while sensing MI quantifies the channel information acquired by the system \cite{ref25,ref26}. \cite{ref27,ref28,ref29} interpreted sensing as a non-cooperative joint source-channel coding process. The above works either assume a single user or that all users can be served. Nevertheless, practical scenarios involve a large number of users, which necessitates the procedure of user scheduling. This further gives rise to a mixed-integer nonlinear programming (MINLP) problem, yet for ISAC-specific scenarios, an efficient algorithm remains underdeveloped. Motivated by the above challenges, the main contributions of this work are summarized as follows:
	\begin{itemize}
		\item We first propose a novel air-ground OmniSteering antenna structure for low-altitude sensing applications. In this design, an omni-steering plate takes the place of the conventional backlobe reflector behind the antenna array. Built on an EM-transparent substrate embedded with passive elements, this plate modulates symmetric backward radiation to arbitrary directions, achieving simultaneous air-ground full-space coverage services. More importantly, by fabricating the omni-steering plate within the antenna module, rather than being deployed as a separate external component like conventional reconfigurable surfaces, the design effectively mitigates three key system drawbacks: first, it minimizes power loss that would otherwise occur when signals interact with external surfaces; second, it eliminates tedious synchronization steps, such as real-time phase matching, that often introduce latency and error risks in external deployment schemes; third, it reduces overhead of inter-device coordination and hardware cost of external mounting structures.
		\item Focusing on a practical dense communication scenario, where the number of communication users exceeds the scheduling capability of the ISAC BS, we consider a joint user scheduling and air-ground omnidirectional S\&C beamforming problem. To be specific, the communication and sensing performance are quantified by sum communication-sensing mutual information (sum-MI), which avoids inconsistencies in the physical unit of S\&C. Observing the separability nature of the sum-MI maximization problem, we judiciously decompose it into two subproblems: one for optimizing the user scheduling vector and active beamforming matrix, and the other for optimizing the passive coefficient matrix. Subsequently, a Riemannian gradient ascent (RGA) algorithm and a semidefinite relaxation (SDR) algorithm are proposed for the passive coefficient optimization subproblem, where RGA achieves comparable performance at lower complexity. For the joint user scheduling and active beamforming subproblem, we analyze the Lagrangian and first-order optimality conditions to equivalently reformulate the original sum-MI maximization problem as a sum weighted minimum mean square error (MMSE) problem, under the impact of user scheduling related variables, finally enabling closed-form updates for combiners and MSE weights, along with convex optimization for active beamforming and scheduling variables. 
		\item Numerical results validate the effectiveness of the proposed algorithms, showing consistent performance gains over baselines in both sum-MI and sum-NMSE performance. Beampattern visualizations further illustrate the system’s ability to align beams with scheduled users and target. Moreover, sensing-MI comparisons confirm that the proposed design significantly enhances coverage while ensuring stable sensing performance.
	\end{itemize}

	The rest of the paper is organized as follows. Section~\ref{sec:sys} depicts the system model. Section~\ref{sec:algorithm} formulates the optimization problem for sum-MI maximization and proposes the US-AGO algorithm to solve it. Simulation results and conclusion are drawn in section~\ref{sec:results} and section~\ref{sec:conclusion}.

	\textit{Notations:} Matrices and vectors are denoted by bold uppercase and lowercase letters, respectively. $[\cdot]^T$, $[\cdot]^H$, and $[\cdot]^*$ stand for the transpose, the Hermitian transpose, and the complex conjugate, respectively. $\mathcal{CN}$ denotes the circularly-symmetric complex Gaussian distribution. $\mathbb{C}(\cdot)$ denotes the complex space. $\mathbb{E}[\cdot]$ denotes statistical expectation. $\det(\cdot)$ and $\operatorname{tr}(\cdot)$ represent matrix determinant and trace, respectively.
	$\operatorname{diag}(\boldsymbol{s})$ is a diagonal matrix with entries from $\boldsymbol{s}$.
	$\|\cdot\|_{2}$ denotes the $\ell_2$-norm.
	$\operatorname{rank}(\cdot)$ denotes matrix rank.
	$\Re(\cdot)$ extracts the real part of a complex number.
	$\odot$ denotes the Hadamard product.

	\section{System Model}\label{sec:sys}
	We consider a monostatic MIMO ISAC system consisting of a BS, $K$ single-antenna communication users (CUs), and a point-like sensing target. To reflect practical deployment scenarios, we focus on a dense communication environment where the number of users exceeds the BS’s scheduling capability, i.e., $K>N_{S}$ with $N_{S}$ denoting the maximum number of scheduled users.
	\begin{figure}[ht]
		\centering
		\includegraphics[width=0.98\columnwidth]{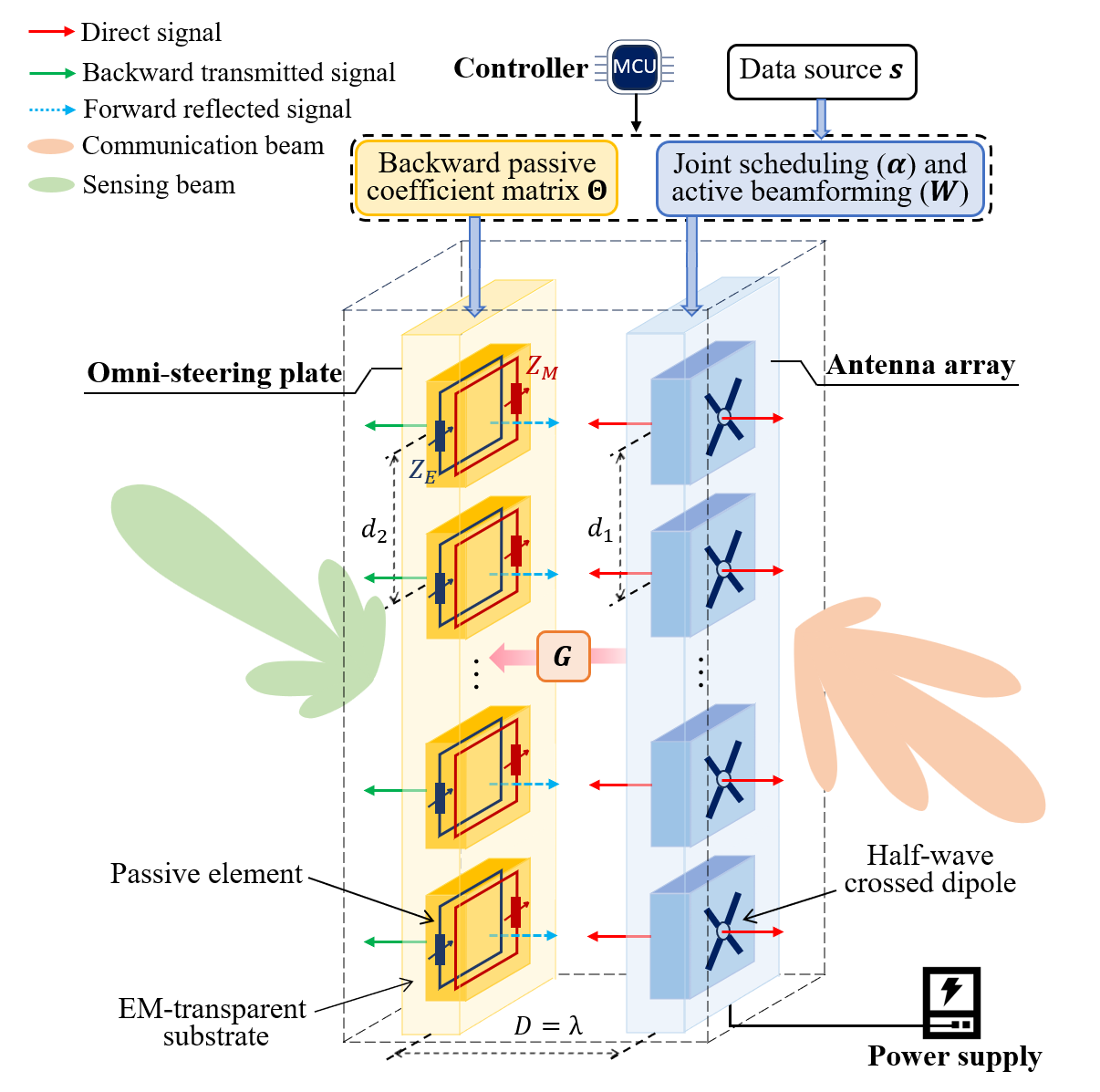} 
		\caption{
			Illustration of the air-ground OmniSteering antenna structure.
		}
		\label{fig2}
		\vspace{-0.5cm} 
	\end{figure}
	
	\vspace{-0.05cm}
	\subsection{Air-Ground OmniSteering Antenna Design} \label{subsec:Mode}
	\textit{(a) Omni-Steering Plate:} To address the beam steering limitations of traditional BS antennas, we propose an novel antenna architecture tailored for air-ground ISAC services. As shown in Fig.~\ref{fig2}, the antenna array, consisting of $N_{T}$ half-wave crossed dipoles, serves as the primary radiating component, directly transmitting signals into the surroundings. 
	
	Importantly, an omni-steering plate replaces the traditional backlobe reflector located behind the antenna array, and operates in a passive manner to further manipulate the symmetric backward signals from the antenna array. In details, the omni-steering plate is constructed on an electromagnetically (EM) transparent substrate embedded with $M$ passive elements. The substrate enables efficient transmission of EM waves with minimal attenuation and can be realized using low-loss dielectric materials, e.g., Rogers RT/Duroid 5880 (relative permittivity = 2.2, tangent loss = 0.0009) \cite{ref31}. Each passive element comprises a parallel LC resonant circuit with metallic loops \cite{ref32}, providing the desired surface electric impedance $Z_{E}$ and magnetic impedance $Z_{M}$. Based on the electromagnetic equivalence principle \cite{ref41}, these impedances determine the surface electric and magnetic currents excited by the incident wave, which satisfy the field boundary conditions and allow arbitrary EM fields on both sides of the surface. By independently tuning $(Z_{E}, Z_{M})$ through integrated varactors under different bias voltages, the plate can dynamically control the transmission and reflection coefficients of incident EM waves, thereby enabling flexible beam reconfiguration.

    \textit{(b) Array-to-Plate Channel:}
	For analytical tractability, we assume that both the antenna array and the omni-steering plate are configured as uniform linear arrays (ULAs) with aligned elements, i.e., $M = N_{T}$. The inter-element spacing of both is set to $d_{1}=d_{2}=d=\lambda/2$, where $\lambda$ denotes the carrier wavelength. The plate is placed at a distance of $D=\lambda$ behind the antenna array, which is less than the Rayleigh distance \cite{ref42}. Then, the near-field channel $\boldsymbol{G}\in\mathbb{C}^{M\times N_T}$ from the active array to the passive omni-steering plate is modeled using a spherical wave-based line-of-sight (LOS) channel model \cite{ref44,ref45}. Specifically, the $(m,n)$-th element of $\boldsymbol{G}$ is
	\begin{equation}
		[\boldsymbol{G}]_{m,n}=\frac{\lambda}{4\pi \sqrt{D^2+((m-n)d)^2}}e^{-i2\pi\sqrt{D^2+((m-n)d)^2}/\lambda}.\label{eq:3}
	\end{equation}

	\textit{(c) System Control and Power Supply:} Building on the structural design above, the hardware control of the air-ground OmniSteering antenna system is realized through a microcontroller unit (MCU) based controller \cite{ref43}, which coordinates the active antenna array and the passive omni-steering plate. The antenna array, powered by an external supply, radiates forward signals to form high-gain communication beams directed toward ground users. Meanwhile, the omni-steering plate reconfigures the symmetric backward signals through controlled transmission or reflection, steering sensing beams toward air targets.	
	
	Next, we characterize the signal manipulation performed by each passive element of the omni-steering plate in terms of its transmission and reflection coefficients. Denote the transmission and reflection coefficients of the $m$-th passive element as $\theta_m^t=\sqrt{\beta_m^t}e^{j\varphi_m^t}$ and $\theta_m^r=\sqrt{\beta_m^r}e^{j\varphi_m^r}$, respectively, where $\sqrt{\beta_{m}^{t}}$, $\sqrt{\beta_{m}^{r}}\in[0{,}1]$ represent the amplitude and $\varphi_{m}^{t}$, $\varphi_{m}^{r}\in[0{,}2\pi)$ represent the phase shift of transmission and reflection, respectively. In this work, all elements are configured either in a full transmission mode (i.e., $\beta_m^t=1$, $\beta_m^r=0, \forall m\in\mathcal{M}$, referred to as T mode) or a full reflection mode (i.e., $\beta_m^t=0$, $\beta_m^r=1, \forall m\in\mathcal{M}$, referred to as R mode). Accordingly, the coefficient vector $\boldsymbol{\theta}$ is given by
	\begin{equation}\boldsymbol{\theta}=
		\begin{cases}
			[\theta_1^t,\theta_2^t,\cdots,\theta_M^t]^T=\left[e^{j\varphi_1^t},e^{j\varphi_2^t},\cdots,e^{j\varphi_M^t}\right]^T,\text{T mode} \\
			[\theta_1^r,\theta_2^r,\cdots,\theta_M^r]
			^T=\left[e^{j\varphi_1^r},e^{j\varphi_2^r},\cdots,e^{j\varphi_M^r}\right]^T,\text{R mode} 
			\label{eq:1}
		\end{cases}
	\end{equation}
	Here, $\boldsymbol{\theta}$ is also referred to as the backward passive coefficient vector in the following discussion.
	
	Note that the above description focuses on transmit antenna structure, while the BS is also equipped with a receive antenna for uplink reception and echo collection. This array is modeled as a ULA with $N_{R}$ elements and $\lambda/2$ inter-element spacing.

	\subsection{S\&C Channel Model}
	We use the Saleh-Valenzuela model to characterize the MIMO downlink channel between the BS and each CU. This model captures the clustered and sparse nature of high-frequency wireless channels, comprising $N_{cl}$ scattering clusters, each containing $N_{ray}$ multipath rays. Let $\boldsymbol{h}_{C,k}^H\in\mathbb{C}^{1\times N_T}$ denote the downlink channel from the BS to the $k$-th CU
	\begin{equation}
		\boldsymbol{h}_{C,k}^H=\sqrt{\frac{N_T}{N_{cl}N_{ray}}}\sum_{c=1}^{N_{cl}}\sum_{l=1}^{N_{ray}}\xi_{c,l}\boldsymbol{a}_T^H(\psi_{c,l}),\label{eq:2}
	\end{equation}
	where $\xi_{c,l}$ is the complex gain of the $l$-th ray in the $c$-th cluster, assumed to follow an i.i.d. complex Gaussian distribution. $\boldsymbol{a}_T(\psi_{c,l})$ represents the normalized transmit array response vector at the BS corresponding to the angle of departure (AoD) $\psi_{c,l}$. Since each CU is equipped with a single antenna, the receive array response reduces to a scalar and is omitted here.
	
	In addition to the communication channel, we model the sensing channel associated with target reflection. Specifically, the target response matrix to be estimated is denoted as $\boldsymbol{H}_{S}\in\mathbb{C}^{N_{R}\times M}$ and can be expressed as
	\begin{equation}
		\boldsymbol{H}_{S} = \xi\boldsymbol{b}_{R}(\psi)\boldsymbol{b}_{T}^H(\psi),
	\end{equation}
	where $\xi$ represents the complex reflection coefficient of the target, while $\boldsymbol{b}_{T}(\psi)\in\mathbb{C}^{M\times 1}$ and $\boldsymbol{b}_{R}(\psi)\in\mathbb{C}^{N_{R}\times 1}$ represent the transmit and receive steering vectors corresponding to the angle $\psi$, respectively. Without loss of generality, we assume that the columns of $\boldsymbol{H}_{S}^{H}$ are i.i.d. and each follows the distribution $\mathcal{CN}(\boldsymbol{0},\boldsymbol{R}_H)$.

	\vspace{-0.1cm}
	\subsection{Signal Model}
	Denote $\boldsymbol{s}=[s_{1},s_{2},\cdots,s_{K}]^{T}\in\mathbb{C}^{K\times1}$ as the ISAC signal vector, where $s_{k}$ is the data symbol desired by the $k$-th CU. Each entry in $\boldsymbol{s}$ is assumed to be statistically orthogonal to each other, so we have $\mathbb{E}[\boldsymbol{s}\boldsymbol{s}^H]=\boldsymbol{I}_K$. Denote $\boldsymbol{\alpha}=[\alpha_{1},\alpha_{2},\cdots,\alpha_{K}]^{T}\in\mathbb{C}^{K\times1}$ as the user scheduling vector, where $\alpha_{k}=\{0{,}1\}$ is the binary variable indicating whether the $k$-th CU is scheduled ($\alpha_{k}=1$) or not ($\alpha_{k}=0$). The active beamforming matrix at the transmit antenna array is denoted by $\boldsymbol{W}=[\boldsymbol{w}_{1},\boldsymbol{w}_{2},\cdots,\boldsymbol{w}_{K}]\in$ $\mathbb{C}^{N_T\times K}$, where $\boldsymbol{w}_k\in\mathbb{C}^{N_T\times1}$ represents the active beamforming vector corresponding to CU $k$. Then, the received signal at CU $k$ can be expressed as
	\begin{equation}
		y_{C,k}=\underbrace{\frac{1}{2}\alpha_{k}\boldsymbol{h}_{C,k}^{H}\boldsymbol{w}_{k}s_{k}}_{\text{desired signal of user }k}+\underbrace{\frac{1}{2}\sum_{j\neq k}^{K}\alpha_{j}\boldsymbol{h}_{C,k}^{H}\boldsymbol{w}_{j}s_{j}}_{\text{multi-user interference}}+z_{C,k},
		\label{eq:4}
	\end{equation}
	where $z_{C,k}\sim\mathcal{CN}(0,\sigma_{C}^{2})$ is additive white Gaussian noise. Note that the factor $\frac{1}{2}$ reflects the power split between the forward and backward lobes. By collecting the received signals from all users, \eqref{eq:4} can be  written in a compact form as
	\begin{equation}
		\boldsymbol{y}_{C}=\frac{1}{2}\boldsymbol{H}_{C}\boldsymbol{W}\mathrm{diag}(\boldsymbol{\alpha})\boldsymbol{s}+\boldsymbol{z}_{C},
		\label{eq:5}
	\end{equation}
	where $\boldsymbol{y}_{C}=\left[y_{C,1},y_{C,2},\cdots,y_{C,K}\right]^T\in\mathbb{C}^{K\times1}$, $\boldsymbol{H}_{C}=\left[\boldsymbol{h}_{C,1},\boldsymbol{h}_{C,2},\cdots,\boldsymbol{h}_{C,K}\right]^H\in\mathbb{C}^{K\times N_T}$, and $\boldsymbol{z}_{C}=\left[z_{C,1},z_{C,2},\cdots,z_{C,K}\right]^T\in\mathbb{C}^{K\times1}$.
	
	Meanwhile, the antenna array also generates symmetric backward signals that arrive at the omni-steering plate, which can be represented as $\frac{1}{2} \boldsymbol{G}\boldsymbol{W}\mathrm{diag}(\boldsymbol{\alpha})\boldsymbol{s}$, and are subsequently reconfigured by the plate for air sensing. The echoed sensing signal received at the BS can then be modeled as
	\begin{equation}
		\boldsymbol{y}_{S}=\frac{1}{2}\boldsymbol{H}_{S}\boldsymbol{\Theta} \boldsymbol{G}\boldsymbol{W}\mathrm{diag}(\boldsymbol{\alpha})\boldsymbol{s}+\boldsymbol{z}_{S},
		\label{eq:6}
	\end{equation}
	where $\boldsymbol{\Theta}=\mathrm{diag}(\boldsymbol{\theta})\in\mathbb{C}^{M\times M}$ is the backward passive coefficient matrix, and $\boldsymbol{z}_{S}\sim\mathcal{CN}(\mathbf{0},\sigma_{S}^{2}\boldsymbol{I}_{N_{R}})$ is additive white Gaussian noise. For simplicity, we define an equivalent sensing beamformer as $\boldsymbol{V}_S\triangleq\boldsymbol{\Theta}\boldsymbol{G}\boldsymbol{W}\mathrm{diag}(\boldsymbol{\alpha})\in\mathbb{C}^{M\times K}$. By applying conjugate transposition, the sensing model in \eqref{eq:6} can be rewritten as
	\begin{equation}
		\boldsymbol{y}_{S}^{H}=\frac{1}{2}\boldsymbol{s}^{H}\boldsymbol{V}_{S}^{H}\boldsymbol{H}_{S}^{H}+\boldsymbol{z}_{S}^{H}.
		\label{eq:7}
	\end{equation}
	
	By comparing this equivalent form with the communication model in \eqref{eq:5}, we obtain an intuitive interpretation of the sensing process. It can be viewed as a virtual uplink in which the target acts as a “transmitter”, delivering the reflected signal $\boldsymbol{H}_{S}^{H}$ back to the BS. In this view, $\boldsymbol{V}_{S}^{H}$ serves as a precoding matrix, while the communication symbol vector $\boldsymbol{s}^{H}$ functions as a known “channel” at BS. This interpretation reveals a duality between communication and sensing.

	\section{Joint User Scheduling and Air-Ground Omnidirectional Sensing and Communication Algorithm Design} \label{sec:algorithm}
	In this section, we first formulate a unified optimization problem that aims to maximize the sum-MI, encompassing both communication and sensing performance. We then propose a comprehensive algorithm to jointly optimize the associated variables. 
	
	\subsection{Problem Formulation}
	The communication MI is defined as the MI between the transmitted symbol $s_{k}$ and received signal $y_{C,k}$ given the knowledge of the channel $\boldsymbol{h}_{C,k}^H$. Based on \eqref{eq:4}, the MI of the $k$-th CU can be expressed as
	\begin{equation}\begin{aligned}
			I_{C,k} & =I\left(y_{C,k};s_{k}|\boldsymbol{h}_{C,k}^{H}\right) \\
			& =\log_2\left(1+\frac{\left|\alpha_k\boldsymbol{h}_{C,k}^H\boldsymbol{w}_k\right|^2}{\sum_{j=1,j\neq k}^K\left|\alpha_j\boldsymbol{h}_{C,k}^H\boldsymbol{w}_j\right|^2+4\sigma_C^2}\right).
			\label{eq:8}
	\end{aligned}\end{equation}

	The sensing MI quantifies the information contained in the received echo signal
	$\boldsymbol{y}_{S}$ about the target response matrix $\boldsymbol{H}_{S}$, given that the transmitted signal $\boldsymbol{s}$ is known at the BS. Following the virtual uplink interpretation introduced earlier, the sensing MI can be derived as 
	\begin{equation}\begin{aligned}
			I_{S} & =I(\boldsymbol{y}_{S};\boldsymbol{H}_{S}|\boldsymbol{s}) \\
			& =I(\boldsymbol{y}_{S}^{H};\boldsymbol{H}_{S}^{H}|\boldsymbol{s}^{H}) \\
			&=\log\det\left(\boldsymbol{I}_{M}+\frac{1}{4N_{R}\sigma_{S}^{2}}\boldsymbol{V}_{S}\boldsymbol{s}\boldsymbol{s}^{H}\boldsymbol{V}_{S}^{H}\boldsymbol{R}_{H}\right) \\
			&  =\log_{2}\left(1+\frac{1}{4N_{R}\sigma_{S}^{2}}\boldsymbol{s}^{H}\boldsymbol{V}_{S}^{H}\boldsymbol{R}_{H}\boldsymbol{V}_{S}\boldsymbol{s}\right),
			\label{eq:9}
	\end{aligned}\end{equation}
	where the last equality holds due to the matrix determinant lemma, i.e.,
	$\det(\boldsymbol{I}+\boldsymbol{u}\boldsymbol{v}^{H})=1+\boldsymbol{v}^{H}\boldsymbol{u}$.
	
	Our objective is to jointly design the user scheduling vector $\boldsymbol{\alpha}$, the active beamforming matrix $\boldsymbol{W}$, and the backward passive coefficient matrix $\boldsymbol{\Theta}$. The optimization problem can be formulated as follows
	\begin{subequations}\label{eq:10}
		\begin{align}
			(\text{P}1)\!:\! \max_{\boldsymbol{\alpha},\boldsymbol{W},\boldsymbol{\Theta}}\, & \kappa \sum_{k=1}^{K}\log_2\left(\!1\!+\frac{\left|\alpha_k\boldsymbol{h}_{C,k}^H\boldsymbol{w}_k\right|^2}{\sum_{j=1,j\neq k}^K\left|\alpha_j\boldsymbol{h}_{C,k}^H\boldsymbol{w}_j\right|^2\!+4\sigma_C^2}\right)\notag\\
			&+(1-\kappa)\log_{2}\left(1+\frac{1}{4N_{R}\sigma_{S}^{2}}\boldsymbol{s}^{H}\boldsymbol{V}_{S}^{H}\boldsymbol{R}_{H}\boldsymbol{V}_{S}\boldsymbol{s}\right) \label{eq10:Za} \\
			\mathrm{s.t.} \;& \sum_{k=1}^{K}\alpha_{k}\|\boldsymbol{w}_{k}\|_{2}^{2} \leq P_{T}, \label{eq10:Zb}\\ 
			& \sum_{k=1}^{K}\alpha_{k} \leq N_{S}, \label{eq10:Zc} \\
			& \alpha_{k}=\{0{,}1\}, \forall k,\label{eq10:Zd} \\
			& |\theta_{m}|=1, \forall m. \label{eq10:Ze}
	\end{align}\end{subequations}
	Here, we use $\theta_m$ to denote the coefficient of the $m$-th passive element on the omni-steering plate, which can refer to either $\theta_m^t$ or $\theta_m^r$ depending on the operating mode. In the formulated problem, \eqref{eq10:Zb} ensures the total transmit power constraint does not exceed the budget $P_{T}$; \eqref{eq10:Zc} limits the number of scheduled users; \eqref{eq10:Zd} specifies the binary scheduling variables; and \eqref{eq10:Ze} is the unit-modulus constraint for the backward passive coefficient. The weighting factor $\kappa \in [0,1]$ is introduced to balance the S\&C objectives in \eqref{eq10:Za}. (\text{P}1) is challenging to solve due to the complicated objective with fractional structure, along with the coupled variables in \eqref{eq10:Za} and \eqref{eq10:Zb}. Additionally, the integer constraint \eqref{eq10:Zd} and the unit-modulus constraint \eqref{eq10:Ze} further exacerbate its non-convexity.
	
	Observing the structures of the problem, it is found that the sensing MI $I_{S}$ is a function with respect to $\boldsymbol{\Theta}$, $\boldsymbol{\alpha}$ and $\boldsymbol{W}$, while communication MI $I_{C,k}$ depends only on $\boldsymbol{\alpha}$ and $\boldsymbol{W}$. Moreover, the constraint involving $\boldsymbol{\Theta}$ in \eqref{eq10:Ze} is independent of the other constraints. The nature of the separability of our problem motivates us to decompose (\text{P}1) into the following two subproblems, the backward passive coefficient optimization subproblem and the joint scheduling and active beamforming optimization subproblem.

	\subsection{Backward Passive Coefficient Optimization} \label{subsec:SDR}
	In this subsection, we optimize the backward passive coefficient matrix $\boldsymbol{\Theta}$ with $\boldsymbol{\alpha}$ and $\boldsymbol{W}$ fixed. The corresponding subproblem formulation is given by
	\begin{subequations}\label{eq:11}
		\begin{align}
			(\text{P}2):\max_{\boldsymbol{\Theta}} \;& \log_{2}\left(1+\frac{1}{4N_{R}\sigma_{S}^{2}}\boldsymbol{s}^{H}\boldsymbol{V}_{S}^{H}\boldsymbol{R}_{H}\boldsymbol{V}_{S}\boldsymbol{s}\right) \label{eq11:Za} \\
			\mathrm{s.t.}\; & |\theta_{m}|=1, \forall m. \label{eq11:Zb}
		\end{align}
	\end{subequations}
	
	To simplify (\text{P}2), we establish the following proposition.
	
	\newtheorem{proposition}{Proposition}
	\newtheorem{proof}{Proof of Proposition}
	\begin{proposition}\label{prop:pro1}
		Let $\boldsymbol{b} \triangleq \boldsymbol{G} \boldsymbol{W} \mathrm{diag}(\boldsymbol{\alpha}) \boldsymbol{s} \in \mathbb{C}^{M \times 1}$, and define $\boldsymbol{C} \triangleq \mathrm{diag}(\boldsymbol{b})^H \boldsymbol{R}_H \mathrm{diag}(\boldsymbol{b}) \in \mathbb{C}^{M \times M}$.
		Then (\text{P}2) can be equivalently reformulated as
		\begin{subequations}\label{eq:12}
			\begin{align}
				\max_{\boldsymbol{\theta}} \; & \boldsymbol{\theta}^H \boldsymbol{C} \boldsymbol{\theta} \label{eq12:Za} \\
				\mathrm{s.t.} \; & |\theta_m| = 1, \forall m. \label{eq12:Zb}
			\end{align}
		\end{subequations}
		\hfill $\square$
	\end{proposition}
	
	\begin{proof}\label{proof:proof1}
		Recall that $\boldsymbol{V}_S = \boldsymbol{\Theta} \boldsymbol{G} \boldsymbol{W} \mathrm{diag}(\boldsymbol{\alpha})$ and $\boldsymbol{b} = \boldsymbol{G} \boldsymbol{W} \mathrm{diag}(\boldsymbol{\alpha}) \boldsymbol{s}$. Substituting into (\text{P}2), the objective function becomes
		\begin{equation}
			\begin{aligned}
				\log_{2}\left( 1 + \frac{1}{4N_R \sigma_S^2} \, \boldsymbol{b}^H \boldsymbol{\Theta}^H \boldsymbol{R}_H \boldsymbol{\Theta} \boldsymbol{b} \right).
			\end{aligned}
		\end{equation}
		We have $\boldsymbol{\Theta} = \mathrm{diag}(\boldsymbol{\theta})$, it yields
		\begin{equation}
			\begin{aligned}
				\boldsymbol{b}^H \boldsymbol{\Theta}^H \boldsymbol{R}_H \boldsymbol{\Theta} \boldsymbol{b} 
				= \left(\mathrm{diag}(\boldsymbol{\theta}) \boldsymbol{b}\right)^H \boldsymbol{R}_H \left(\mathrm{diag}(\boldsymbol{\theta}) \boldsymbol{b}\right).
			\end{aligned}
		\end{equation}
		Using the identity $\mathrm{diag}(\boldsymbol{X}) \boldsymbol{Y}=\mathrm{diag}(\boldsymbol{Y}) \boldsymbol{X}$, we further rewrite it as
		\begin{equation}
			\begin{aligned}
				\left(\mathrm{diag}(\boldsymbol{b}) \boldsymbol{\theta}\right)^H \boldsymbol{R}_H \left(\mathrm{diag}(\boldsymbol{b}) \boldsymbol{\theta}\right)
				= \boldsymbol{\theta}^H \left(\mathrm{diag}(\boldsymbol{b})^H \boldsymbol{R}_H \mathrm{diag}(\boldsymbol{b})\right) \boldsymbol{\theta}.
			\end{aligned}
		\end{equation}
		Defining $\boldsymbol{C}\triangleq\mathrm{diag}(\boldsymbol{b})^H\boldsymbol{R}_{H}\mathrm{diag}(\boldsymbol{b})$ yields the equivalent formulation in (\text{P}2). This completes the proof. \hfill $\blacksquare$
	\end{proof}
	
	With Proposition 1, (P2) is reduced to a quadratic optimization problem with unit-modulus constraints as in \eqref{eq:12}. Next, we present two approaches to tackle this problem efficiently.
	
	\textit{(1) RGA-Based Coefficient Design:}
	We first propose a Riemannian gradient ascent based coefficient design, which can directly handle the unit-modulus constraints with low complexity \cite{ref14,ref33}. By viewing the feasible set of \eqref{eq:12} as a complex circle manifold, the problem can be equivalently rewritten in the standard form for manifold optimization
	\begin{equation}
		\max_{\boldsymbol{\theta}\in\mathcal{S}}f(\boldsymbol{\theta}),\label{eq:13}
	\end{equation}
	where $\mathcal{S}=\{\boldsymbol{\theta}\in\mathbb{C}^{M\times1}{:}|\theta_1|=|\theta_2|=\cdots=|\theta_M|=1\}$ denotes the complex circle manifold, and $f(\boldsymbol{\theta})=\boldsymbol{\theta}^H\boldsymbol{C}\boldsymbol{\theta}$.
	
	Similar to gradient ascent methods in Euclidean space, the RGA algorithm consists of computing an ascent direction within the tangent space of the manifold and updating the iterate along this direction. Specifically, the tangent space at any point $\boldsymbol{\theta} \in \mathcal{S}$ is defined as
	\begin{equation}\label{eq:14}
		T_{\boldsymbol{\theta}}\mathcal{S} \triangleq \{\mathbf{z} \in \mathbb{C}^M \mid \Re\{\mathbf{z} \odot \boldsymbol{\theta}^H\} = \mathbf{0}_M\}.
	\end{equation}
	which contains all feasible directions that preserve the unit-modulus constraint to first order. The Euclidean gradient of $f(\boldsymbol{\theta})$ is given by
	\begin{equation}\label{eq:15}
		\nabla_{\boldsymbol{\theta}} f=2\boldsymbol{C}\boldsymbol{\theta}.
	\end{equation}
	To obtain the Riemannian gradient, the Euclidean gradient is projected onto the tangent space using the projection operator
	\begin{equation}\label{eq:16}
		\mathcal{P}_{T_{\boldsymbol{\theta}}\mathcal{S}}(\mathbf{z})\triangleq\mathbf{z}-\Re\{\mathbf{z}\odot\boldsymbol{\theta}^H\}\odot\boldsymbol{\theta}.
	\end{equation}
	Thus, the Riemannian gradient is
	\begin{equation}\label{eq:17}
		\mathrm{grad}_{\boldsymbol{\theta}}f\triangleq\nabla_{\boldsymbol{\theta}}f-\Re\{\nabla_{\boldsymbol{\theta}}f\odot\boldsymbol{\theta}^{H}\}\odot\boldsymbol{\theta}.
	\end{equation}
	
	To improve convergence speed, we employ the Polak–Ribiére conjugate gradient algorithm, which incorporates quasi-second-order information by combining the current Riemannian gradient with the previous update direction. The update direction $\boldsymbol{\eta}_q$ at the $q$-th iteration is computed as
	\begin{equation} \label{eq:18}
		\boldsymbol{\eta}_q=\operatorname{grad}_{\boldsymbol{\theta}^{(q)}}f+\gamma_qT_{\boldsymbol{\theta}^{(q-1)}\to\boldsymbol{\theta}^{(q)}}(\boldsymbol{\eta}_{q-1}),
	\end{equation}
	where $\gamma_q$ is the Polak–Ribiére parameter, and $T_{\boldsymbol{\theta}^{(q-1)}\to\boldsymbol{\theta}^{(q)}}(\boldsymbol{\eta}_{q-1})$ is the vector transport operator projecting the previous search direction $\boldsymbol{\eta}_{q-1}$ onto the current tangent space, defined as
	\begin{equation}\label{eq:19}
		\begin{aligned}
			T_{\boldsymbol{\theta}^{(q-1)}\to\boldsymbol{\theta}^{(q)}}(\boldsymbol{\eta}_{q-1}) & \triangleq T_{\boldsymbol{\theta}^{(q-1)}}\mathcal{S}\to T_{\boldsymbol{\theta}^{(q)}}\mathcal{S} \\
			& \triangleq\boldsymbol{\eta}_{q-1}-\Re\{\boldsymbol{\eta}_{q-1}\odot(\boldsymbol{\theta}^{(q)})^H\}\odot\boldsymbol{\theta}^{(q)}.
		\end{aligned}
	\end{equation}
	The Polak-Ribiére parameter $\gamma_q$ is given by
	\begin{equation}\label{eq:20}
		\gamma_q=\frac{\langle\mathrm{grad}_{\boldsymbol{\theta}^{(q)}}f,\boldsymbol{J}_q\rangle}{\langle\mathrm{grad}_{\boldsymbol{\theta}^{(q-1)}}f,\mathrm{grad}_{\boldsymbol{\theta}^{(q-1)}}f\rangle},
	\end{equation}
	where $	\boldsymbol{J}_q\triangleq\operatorname{grad}_{\boldsymbol{\theta}^{(q)}}f-T_{\boldsymbol{\theta}^{(q-1)}\to\boldsymbol{\theta}^{(q)}}(\boldsymbol{\eta}_{q-1})$ is the difference between the current gradient and the transported previous direction. In essence, $\gamma_q$ quantifies their correlation and thus captures curvature information, enabling the algorithm to refine the update direction beyond first-order methods.

	Once the ascent direction is determined, the next iterate is obtained via the retraction 
	\begin{equation}\label{eq:22}
		\boldsymbol{\theta}^{(q+1)}=\mathcal{R}_{\boldsymbol{\theta}^{(q)}}(\mu_{q}\boldsymbol{\eta}_{q}),
	\end{equation}
	which maps the updated point back onto the complex circle manifold while preserving unit-modulus constraints. Here, $\mu_{q}$ is the step size, typically computed by a line search method such as the Armijo rule. The retraction operator $\mathcal{R}_{\boldsymbol{\theta}}(\cdot)$ for the complex circle manifold is defined as
	\begin{equation} \label{eq:23}
		\mathcal{R}_{\boldsymbol{\theta}}(\mathbf{z})\triangleq\left(\frac{\mathsf{z}_1}{|\mathsf{z}_1|},\cdots,\frac{\mathsf{z}_M}{|\mathsf{z}_M|}\right)^T.
	\end{equation}
	
	The overall procedure of the RGA algorithm for optimizing backward passive coefficient matrix $\boldsymbol{\Theta}$ is summarized in Algorithm \autoref{alg:1}.
	
	\begin{algorithm}[htbp]
	\caption{RGA algorithm for Backward Passive Coefficient Optimization} \label{alg:1}
	\begin{algorithmic}[1]
		\State \textbf{Input:} $\boldsymbol{C}$, $q_{\max}$, $\epsilon$.
		\State \textbf{Initialize} $\boldsymbol{\theta}^{(0)}=\boldsymbol{\theta}^{(1)}\in\mathcal{S}$, $\eta_{0}=\mathrm{grad}_{\boldsymbol{\theta}^{(0)}}f$, and $q=1$.
		\State \textbf{Repeat}
		\State Compute the Euclidean gradient $\nabla_{\boldsymbol{\theta}^{(q)}}f$ via \eqref{eq:15}.
		\State Compute the Riemannian gradient $\mathrm{grad}_{\boldsymbol{\theta}^{(q)}}f$ via \eqref{eq:17}.
		\State Compute the Polak–Ribiére parameter $\gamma_{q}$ via \eqref{eq:20}.
		\State Compute the ascent direction $\boldsymbol{\eta}_q$ via \eqref{eq:18}.
		\State Update $\boldsymbol{\theta}^{(q+1)}$ using retraction \eqref{eq:22}, i.e., $\boldsymbol{\theta}^{(q+1)}=\mathcal{R}_{\boldsymbol{\theta}^{(q)}}(\mu_q \boldsymbol{\eta}_q)$, where $\mu_q$ is the step size can be computed by a line search method such as the Armijo rule.
		\State $q=q+1$.
		\State \textbf{Until} $\left\|\mathrm{grad}_{\boldsymbol{\theta}^{(q)}}f\right\|_{2}\leq\epsilon$ or $q\geq{q_\mathrm{max}}$.
		\State \textbf{Output:} $\boldsymbol{\Theta}=\mathrm{diag}(\boldsymbol{\theta}^{(q+1)})$
		\end{algorithmic}
	\end{algorithm}

	\textit{(2) SDR-Based Coefficient Design:} As an alternative, semidefinite relaxation provides a classic approach for handling such non-convex quadratic problems. Specifically, problem \eqref{eq:12} can be rewritten as a rank-constrained semidefinite program by introducing the lifted variable $\boldsymbol{Q} \triangleq \boldsymbol{\theta}\boldsymbol{\theta}^H \in \mathbb{C}^{M \times M}$. Note that $\boldsymbol{Q}$ satisfies $\boldsymbol{Q} \succeq \boldsymbol{0}$ and $\mathrm{rank}(\boldsymbol{Q}) = 1$. Using this transformation, the objective function in \eqref{eq:12} becomes $\operatorname{tr}(\boldsymbol{C}\boldsymbol{Q})$, and the unit-modulus constraints on $\boldsymbol{\theta}$ translate to $[\boldsymbol{Q}]_{m,m}=1$.

	Thus, the original problem \eqref{eq:12} can be equivalently reformulated as follows
	\begin{subequations} \label{eq:24}
		\begin{align}
			\max_{\boldsymbol{Q}}\;& \operatorname{tr}(\boldsymbol{C}\boldsymbol{Q})  \label{eq24:Za} \\
			\mathrm{s.t.}\;&\mathrm{diag}(\boldsymbol{Q}) =\boldsymbol{1},\label{eq24:Zb} \\
			&\boldsymbol{Q}\succeq\boldsymbol{0},\label{eq24:Zc} \\
			&\mathrm{rank}(\boldsymbol{Q}) =1,\label{eq24:Zd}
		\end{align}
	\end{subequations}

	By omitting the rank-one constraint in \eqref{eq24:Zd}, the problem reduces to a standard semidefinite program (SDP), which can be efficiently solved using convex optimization tools such as CVX. However, the optimal solution $\boldsymbol{Q}^\mathrm{opt}$ to the relaxed SDP is generally not of rank one. To address this issue, Gaussian randomization is applied to generate multiple candidate rank-one solutions based on $\boldsymbol{Q}^\mathrm{opt}$.
	\begin{equation}\label{eq:101}
		\tilde{\boldsymbol{\theta}}^{(q)} = e^{i\angle\left(\boldsymbol{\Omega}\boldsymbol{\Sigma}^{1/2}\boldsymbol{x}^{(q)}\right)},
		\quad q=1,\dots, N_{\mathrm{rand}},
	\end{equation}
	where $\boldsymbol{Q}^\mathrm{opt}=\boldsymbol{\Omega}\boldsymbol{\Sigma}\boldsymbol{\Omega}^H$ is the eigenvalue decomposition and $\boldsymbol{x}^{(q)}\sim\mathcal{CN}(\boldsymbol{0},\boldsymbol{I})$ is a standard complex Gaussian random vector. Among these candidates, the final approximate solution $\boldsymbol{\theta}^\mathrm{opt}$ is obtained by selecting the one that yields the highest objective value
	\begin{equation}\label{eq:102}
		\boldsymbol{\theta}^\mathrm{opt} = \arg\max_{\tilde{\boldsymbol{\theta}}^{(q)}}\Re\left\{\operatorname{tr}\left(\boldsymbol{C}\tilde{\boldsymbol{\theta}}^{(q)}(\tilde{\boldsymbol{\theta}}^{(q)})^H\right)\right\}.
	\end{equation}
	
	The SDR algorithm is summarized in Algorithm \autoref{alg:3}. While SDR with Gaussian randomization incurs higher computational complexity than the RGA algorithm, it provides a benchmark solution that approaches the global optimum of the relaxed SDP, offering performance validation and robustness guarantees for the passive coefficient design.

	\subsection{Joint Scheduling and Active Beamforming Optimization}
	In this subsection, we optimize the user scheduling vector $\boldsymbol{\alpha}$ and active beamforming matrix $\boldsymbol{W}$ with $\boldsymbol{\Theta}$ fixed.
	\begin{subequations}\label{eq:25}
		\begin{align}
			(\text{P}3):\max_{\boldsymbol{\alpha},\boldsymbol{W}}\; & \kappa\sum_{k=1}^{K}\log_2\left(1+\frac{\left|\alpha_k\boldsymbol{h}_{C,k}^H\boldsymbol{w}_k\right|^2}{\sum_{j=1,j\neq k}^K\left|\alpha_j\boldsymbol{h}_{C,k}^H\boldsymbol{w}_j\right|^2+4\sigma_C^2}\right)\notag\\
			&+(1-\kappa)\log_{2}\left(1+\frac{1}{4N_{R}\sigma_{S}^{2}}\boldsymbol{s}^{H}\boldsymbol{V}_{S}^{H}\boldsymbol{R}_{H}\boldsymbol{V}_{S}\boldsymbol{s}\right) \label{eq25:Za} \\
			\mathrm{s.t.} \;& \sum_{k=1}^{K}\alpha_{k}\|\boldsymbol{w}_{k}\|_{2}^{2} \leq P_{T}, \label{eq25:Zb}\\ 
			& \sum_{k=1}^{K}\alpha_{k} \leq N_{S}, \label{eq25:Zc} \\
			& \alpha_{k}=\{0{,}1\}, \forall k.\label{eq25:Zd}
	\end{align}\end{subequations}
	
	We begin by decoupling $\boldsymbol{\alpha}$ and $\boldsymbol{W}$ in constraint \eqref{eq25:Zb}. Specifically, we reformulate it as
	\begin{equation}\label{eq:26}
		\sum_{k=1}^{K}\|\boldsymbol{w}_{k}\|_{2}^{2}\leqslant P_{T},\quad\|\boldsymbol{w}_{k}\|_{2}^{2}\leqslant\alpha_{k}P_{T},\forall k.
	\end{equation}
	This is based on the intuitive observation that if CU $k$ is not scheduled, i.e., $\alpha_{k}=0$, its associated precoder $\boldsymbol{w}_{k}=0$, and otherwise. Accordingly, subproblem (\text{P}3) becomes
	\begin{subequations}\label{eq:27}
		\begin{align}
			\max_{\boldsymbol{\alpha},\boldsymbol{W}}\; & \kappa\sum_{k=1}^{K}\log_2\left(1+\frac{\left|\alpha_k\boldsymbol{h}_{C,k}^H\boldsymbol{w}_k\right|^2}{\sum_{j=1,j\neq k}^K\left|\alpha_j\boldsymbol{h}_{C,k}^H\boldsymbol{w}_j\right|^2+4\sigma_C^2}\right)\notag\\
			&+(1-\kappa)\log_{2}\left(1+\frac{1}{4N_{R}\sigma_{S}^{2}}\boldsymbol{s}^{H}\boldsymbol{V}_{S}^{H}\boldsymbol{R}_{H}\boldsymbol{V}_{S}\boldsymbol{s}\right) \label{eq27:Za}\\
			\mathrm{s.t.}\;&\sum_{k=1}^{K}\|\boldsymbol{w}_{k}\|_{2}^{2}\leqslant P_{T},\label{eq27:Zb} \\
			&\|\boldsymbol{w}_{k}\|_{2}^{2}\leqslant\alpha_{k}P_{T},\forall k,\label{eq27:Zc} \\
			& \sum_{k=1}^{K}\alpha_{k}\leqslant N_{S},\label{eq27:Zd} \\
			& \alpha_{k}=\{0,1\},\forall k.\label{eq27:Ze}
	\end{align}\end{subequations}

	\begin{table*}[!hb]
		\normalsize 
		\centering
		\hrule	\smallskip
		\begin{minipage}{\textwidth}
			\begin{equation}\label{eq:31}
				\begin{aligned}
					e_{C,k}& =\mathbb{E}\left[(u_{C,k}y_{C,k}-s_k)(u_{C,k}y_{C,k}-s_k)^H\right]=\left|\frac{1}{2}u_{C,k}\alpha_k\boldsymbol{h}_{C,k}^H\boldsymbol{w}_k-1\right|^2+\sum_{j\neq k}^K\left|\frac{1}{2}u_{C,k}\alpha_j\boldsymbol{h}_{C,k}^H\boldsymbol{w}_j\right|^2+\sigma_C^2\left|u_{C,k}\right|^2
				\end{aligned}
			\end{equation}
			\begin{equation}\label{eq:32}
				\begin{aligned}
					\boldsymbol{E}_S=\;&\mathbb{E}[(\boldsymbol{u}_S\boldsymbol{y}_S^H-\boldsymbol{H}_S^H)(\boldsymbol{u}_S\boldsymbol{y}_S^H-\boldsymbol{H}_S^H)^H]\\
					=\;&\frac{1}{4}\boldsymbol{u}_S\boldsymbol{s}^H\mathrm{diag}(\boldsymbol{\alpha})^H\boldsymbol{W}^H\boldsymbol{G}^H\boldsymbol{\Theta}^H\boldsymbol{R}_H\boldsymbol{\Theta}\boldsymbol{G}\boldsymbol{W}\mathrm{diag}(\boldsymbol{\alpha})\boldsymbol{s}\boldsymbol{u}_S^H-\frac{1}{2}\boldsymbol{u}_S\boldsymbol{s}^H\mathrm{diag}(\boldsymbol{\alpha})^H\boldsymbol{W}^H\boldsymbol{G}^H\boldsymbol{\Theta}^H\boldsymbol{R}_H-\frac{1}{2}\boldsymbol{R}_H\boldsymbol{\Theta}\boldsymbol{G}\boldsymbol{W}\mathrm{diag}(\boldsymbol{\alpha})\boldsymbol{s}\boldsymbol{u}_S^H\\
					&+N_R\boldsymbol{\sigma}_S^2\boldsymbol{u}_S\boldsymbol{u}_S^H+\boldsymbol{R}_H
				\end{aligned}
			\end{equation}
		\end{minipage}
	\end{table*}

	To address the binary constraint in \eqref{eq27:Ze}, we adopt a continuous relaxation strategy by introducing a smooth penalty function for each $\alpha_{k}$, given by 
	\begin{equation}\label{eq:28}
		f_\rho(\alpha_k)=\rho[\alpha_k\log(\alpha_k)+(1-\alpha_k)\log(1-\alpha_k)],
	\end{equation}
	where $\rho$ is a positive constant that determines the penalty weight. This function is convex over $\alpha_k\in[0,1]$ and is able to attain its maximum value (i.e., 0) when $\alpha_k$ is exactly 0 or 1. Thus, maximizing this penalty encourages the solution to remain near the binary boundaries, effectively approaching the original discrete constraint. Based on this relaxation, problem (\text{P}3) is reformulated as
	\begin{subequations}\label{eq:29}
		\begin{align}
			\max_{\boldsymbol{\alpha},\boldsymbol{W}}\; & \kappa\sum_{k=1}^{K}\log_2\left(1+\frac{\left|\alpha_k\boldsymbol{h}_{C,k}^H\boldsymbol{w}_k\right|^2}{\sum_{j=1,j\neq k}^K\left|\alpha_j\boldsymbol{h}_{C,k}^H\boldsymbol{w}_j\right|^2+4\sigma_C^2}\right)\notag\\
			&+(1-\kappa)\log_{2}\left(1+\frac{1}{4N_{R}\sigma_{S}^{2}}\boldsymbol{s}^{H}\boldsymbol{V}_{S}^{H}\boldsymbol{R}_{H}\boldsymbol{V}_{S}\boldsymbol{s}\right)\notag\\
			&+\sum_{k=1}^Kf_\rho(\alpha_k)\label{eq29:Za}\\
			\mathrm{s.t.}\;&\eqref{eq27:Zb}, \eqref{eq27:Zc}, \eqref{eq27:Zd}.\label{eq29:Zb}
		\end{align}
	\end{subequations}

	Now, the objective function remains non-convex. For handling such as quadratic-over-quadratic structure, one may transform it into a SDP problem. However, the variables often shift from vectors to high-dimensional positive semi-definite matrices, leading to high computational complexity. To facilitate low complexity solution, we exploit the relationship between mutual information and minimum mean square error, together with the effect of user scheduling, as stated in the following proposition.
	
	\begin{algorithm}[tbp]
	\caption{SDR Algorithm for Backward Passive Coefficient Optimization} 
	\label{alg:3}
	\begin{algorithmic}[1]
		\State \textbf{Input:} $\boldsymbol{C}$, $N_{\mathrm{rand}}$.
		\State \textbf{Initialize} $q=1$
		\State Solve the relaxed SDP problem \eqref{eq:24} via CVX, obtain the optimal solution $\boldsymbol{Q}^\mathrm{opt}$.
		\State Perform eigenvalue decomposition $\boldsymbol{Q}^\mathrm{opt} = \boldsymbol{\Omega}\boldsymbol{\Sigma}\boldsymbol{\Omega}^H$.
		\State \textbf{Repeat}
		\State Generate $\boldsymbol{x}^{(q)} \sim \mathcal{CN}(\boldsymbol{0}, \boldsymbol{I})$.
		\State Construct candidate solution $\tilde{\boldsymbol{\theta}}^{(q)}$ via \eqref{eq:101}.
		\State $q=q+1$. 
		\State \textbf{Until} $q\geq N_{\mathrm{rand}}$.
		\State Select the best candidate $\boldsymbol{\theta}^\mathrm{opt}$ via \eqref{eq:102}.
		\State \textbf{Output:} $\boldsymbol{\Theta}=\mathrm{diag}(\boldsymbol{\theta}^\mathrm{opt})$.
	\end{algorithmic}
\end{algorithm}

	\begin{proposition} \label{prop:pro2}
		Let auxiliary variables $\{u_{C,k}\}$, $\boldsymbol{u}_{S}$, and $\{\Lambda_{C,k}\}$, $\boldsymbol{\Lambda}_{S}$ represent the receiving combiner and weight for communication and sensing, respectively, the original sum-MI maximization problem can be equivalently transformed into a convex sum weighted MSE minimization problem without loss of optimality. Denote $e_{C,k}$ and $\boldsymbol{E}_S$ as the mean square errors in estimating $s_k$ for CU $k$ and that in estimating $\boldsymbol{H}_S^H$ for the sensing target, respectively, with their explicit expressions given in \eqref{eq:31} and \eqref{eq:32}. Problem \eqref{eq:29} can be rewritten as
		\begin{subequations}\label{eq:30}
			\begin{align} 
				\min_{\substack{\{u_{C,k}\}, \boldsymbol{u}_S, \{\Lambda_{C,k}\},\\ \boldsymbol{\Lambda}_S, \boldsymbol{\alpha}, \boldsymbol{W}}}\,&\kappa\sum_{k=1}^K(\Lambda_{C,k}e_{C,k}-\mathrm{log}\Lambda_{C,k})\notag\\
				&+(1-\kappa)\left(\operatorname{tr}(\boldsymbol{\Lambda}_S\boldsymbol{E}_S)-\mathrm{logdet}(\boldsymbol{\Lambda}_S)\right)\notag\\
				&-\sum_{k=1}^Kf_\rho(\alpha_k) \label{eq30:Za}\\
				\mathrm{s.t.}\;&\eqref{eq27:Zb}, \eqref{eq27:Zc}, \eqref{eq27:Zd},\label{eq30:Zb}
		\end{align}\end{subequations}
		\hfill $\square$
	\end{proposition}

	\begin{proof}\label{proof:proof2}
		The equivalence is established by comparing the first-order optimality conditions of the original sum-MI maximization problem and the proposed sum weighted MMSE formulation.
		
		We begin by compactly denoting the constraints in problem \eqref{eq:29}  as $\boldsymbol{g}(\boldsymbol{\alpha},\boldsymbol{W})\preceq\mathbf{0}$, where the notation “$\preceq$” is understood in an element-wise sense. To facilitate subsequent derivations, we define 
		$D_k\triangleq\left|\alpha_k\boldsymbol{h}_{C,k}^H\boldsymbol{w}_k\right|^2$, $N_k\triangleq\sum_{j=1,j\neq k}^K|\alpha_j\boldsymbol{h}_{C,k}^H\boldsymbol{w}_j|^2+4\sigma_C^2$, $\mathit{\Gamma}\triangleq\frac{1}{4N_{R}\sigma_{S}^{2}}\boldsymbol{s}^{H}\boldsymbol{V}_{S}^{H}\boldsymbol{R}_{H}\boldsymbol{V}_{S}\boldsymbol{s}$.
		With these notations, problem \eqref{eq:29} can be equivalently rewritten as
		\begin{subequations}\label{eq:MI}
			\begin{align}
				\min_{\boldsymbol{\alpha},\boldsymbol{W}}\;&-\kappa\sum_{k=1}^K\log_2\left(1+\frac{D_k}{N_k}\right)-(1-\kappa)\log_2(1+\mathit{\Gamma})\notag\\
				&-\sum_{k=1}^Kf_\rho(\alpha_k)\label{eqMI:Zb}  \\
				\mathrm{s.t.}\;&\boldsymbol{g}(\boldsymbol{\alpha},\boldsymbol{W})\preceq\mathbf{0}.\label{eqMI:Zb} 
			\end{align}
		\end{subequations}
		
			\begin{table*}[!hb]
			\normalsize  
			\centering
			\hrule	\smallskip
			\begin{minipage}{\textwidth}
				\begin{equation}\label{eq:LW1}
					\begin{aligned}
						\nabla_{\boldsymbol{w}_{k}}\mathcal{L}_{(38)} =-\alpha_k^2\kappa\left(\frac{\boldsymbol{h}_{C,k}\boldsymbol{h}_{C,k}^H\boldsymbol{w}_k}{D_k+N_k}-\sum_{j\neq k}\frac{D_j}{N_j(D_j+N_j)}\boldsymbol{h}_{C,j}\boldsymbol{h}_{C,j}^H\boldsymbol{w}_k\right)-\frac{1-\kappa}{1+\mathit{\Gamma}} \cdot\frac{\alpha_{k}s_{k}^{*}}{4N_{R}\sigma_{S}^{2}}\boldsymbol{B}\boldsymbol{W}\mathrm{diag}(\boldsymbol{\alpha})\boldsymbol{s}-\frac{\partial(\boldsymbol{\mu}^{T}\mathbf{g}(\boldsymbol{\alpha},\boldsymbol{W}))}{\partial\boldsymbol{w}_{k}}
				\end{aligned}\end{equation}
				\begin{equation}\label{eq:LA1}
					\nabla_{\alpha_k}\mathcal{L}_{(38)}=-\alpha_k^2\kappa\left(\frac{|\boldsymbol{h}_{C,k}^H\boldsymbol{w}_k|^2}{D_k+N_k}-\sum_{j\neq k}\frac{D_j|\boldsymbol{h}_{C,j}^H\boldsymbol{w}_k|^2}{N_j\left(D_j+N_j\right)}\right)-\frac{1-\kappa}{1+\mathit{\Gamma}}\cdot\frac{s_k^*}{4N_R\sigma_S^2}\boldsymbol{w}_k^H\boldsymbol{B}\boldsymbol{W}\mathrm{diag}(\boldsymbol{\alpha})\boldsymbol{s}-f_\rho^{\prime}(\alpha_k)-\frac{\partial(\boldsymbol{\mu}^T\boldsymbol{g}(\boldsymbol{\alpha},\boldsymbol{W}))}{\partial\alpha_k}
				\end{equation}
				\begin{equation}\label{eq:LW2}
					\begin{aligned}
						\nabla_{\boldsymbol{w}_k}\mathcal{L}_{(47)}=&-\alpha_k^2\kappa\left(\frac{\Lambda_{C,k}N_k\boldsymbol{h}_{C,k}\boldsymbol{h}_{C,k}^H\boldsymbol{w}_k}{(D_k+N_k)^2}-\sum_{j\neq k}\frac{\Lambda_{C,j}D_j\boldsymbol{h}_{C,j}\boldsymbol{h}_{C,j}^H\boldsymbol{w}_k}{(D_j+N_j)^2}\right)-\frac{(1-\kappa)\alpha_{k}s_{k}^{*}}{4N_{R}\sigma_{S}^{2}(1+\mathit{\Gamma})}\boldsymbol{M}^H\boldsymbol{E}_{S}^{\mathrm{MMSE}}\boldsymbol{\Lambda}_S\boldsymbol{R}_H\boldsymbol{M}\boldsymbol{W}\mathrm{diag}(\boldsymbol{\alpha})\boldsymbol{s}\\
						&-\frac{\partial(\boldsymbol{\mu}^{T}\mathbf{g}(\boldsymbol{\alpha},\boldsymbol{W}))}{\partial\boldsymbol{w}_{k}}		
					\end{aligned}
				\end{equation}
				\begin{equation}\label{eq:LA2}
				\begin{aligned}
					\nabla_{\alpha_{k}}\mathcal{L}_{(47)}=&-\alpha_{k}^{2}\kappa\left(\frac{\Lambda_{C,k}N_{k}|\boldsymbol{h}_{C,k}^{H}\boldsymbol{w}_{k}|^{2}}{(D_{k}+N_{k})^{2}}-\sum_{j\neq k}\Lambda_{C,j}\frac{D_{j}|\boldsymbol{h}_{C,j}^{H}\boldsymbol{w}_{k}|^{2}}{(D_{j}+N_{j})^{2}}\right)-\frac{(1-\kappa)s_k^*}{4N_R\sigma_S^2}\boldsymbol{w}_k^H\boldsymbol{M}^H\boldsymbol{E}_{S}^{\mathrm{MMSE}}\boldsymbol{\Lambda}_S\boldsymbol{E}_{S}^{\mathrm{MMSE}}\boldsymbol{M}\boldsymbol{W}\mathrm{diag}(\boldsymbol{\alpha})\boldsymbol{s}\\
					&-f_\rho^{\prime}(\alpha_k)-\frac{\partial(\boldsymbol{\mu}^T\boldsymbol{g}(\boldsymbol{\alpha},\boldsymbol{W}))}{\partial\alpha_k}
					\end{aligned}
				\end{equation}	
			\end{minipage}
			\vspace{-0.5cm} 
		\end{table*}
		Then we apply the MMSE principle and introduce the optimal receive combiners for communication and sensing. For CU $k$, the MMSE combiner \cite{ref34} is
		\begin{equation} \label{eq:37}
			u_{C,k}=\frac{1}{2}\alpha_{k}\boldsymbol{w}_{k}^{H}\boldsymbol{h}_{C,k}(\sum_{j=1}^{k}\frac{1}{4}\alpha_{j}\boldsymbol{h}_{C,k}^{H}\boldsymbol{w}_{j}\boldsymbol{w}_{j}^{H}\boldsymbol{h}_{C,k}+\sigma_{C}^{2})^{-1}.
		\end{equation}
		For the sensing target, the combiner of BS's sensing receiver is given as
		\begin{equation} \label{eq:38}
			\boldsymbol{u}_{S}=2\boldsymbol{R}_{H}\boldsymbol{V}_{S}\boldsymbol{s}(\boldsymbol{s}^{H}\boldsymbol{V}_{S}^{H}\boldsymbol{R}_{H}\boldsymbol{V}_{S}\boldsymbol{s}+4N_{R}\sigma_{S}^{2})^{-1}.
		\end{equation}
		The associated scalar and matrix MMSEs \cite{ref35} are 
		\begin{equation}\label{eq:39}
			\begin{aligned}
				e_{C,k}^\mathrm{MMSE}&=1-\frac{1}{4}\left|\alpha_{k}\boldsymbol{h}_{C,k}^{H}\boldsymbol{w}_{k}\right|^{2}\left(\sum_{j=1}^{k}\frac{1}{4}\left|\alpha_{j}\boldsymbol{h}_{C,k}^{H}\boldsymbol{w}_{j}\right|^{2}+\sigma_{C}^{2}\right)^{-1}\\
				&=(1+\frac{D_k}{N_k})^{-1},\forall k
		\end{aligned}\end{equation}
		\begin{equation}\label{eq:40}
			\begin{aligned}
				\boldsymbol{E}_{S}^{\mathrm{MMSE}}=\boldsymbol{R}_{H}(\boldsymbol{I}_{M}+\frac1{4N_{R}\sigma_{S}^{2}}\boldsymbol{V}_{S}\boldsymbol{s}\boldsymbol{s}^{H}\boldsymbol{V}_{S}^{H}\boldsymbol{R}_{H})^{-1}.
		\end{aligned}\end{equation}
	
		Assuming that the MMSE receive combiners are employed, the sum weighed MMSE problem can be formulated as
		\begin{subequations}\label{eq:MSE}
			\begin{align} 
				\min_{\boldsymbol{\alpha}, \boldsymbol{W}}\;&\kappa\sum_{k=1}^K\Lambda_{C,k}e_{C,k}^\mathrm{MMSE}+(1-\kappa)\operatorname{tr}(\boldsymbol{\Lambda}_S\boldsymbol{E}_{S}^{\mathrm{MMSE}})-\sum_{k=1}^K\left.f_\rho(\alpha_k)\right. \label{eqMSE:Za}\\
				\mathrm{s.t.}\;&\boldsymbol{g}(\boldsymbol{\alpha},\boldsymbol{W})\preceq\mathbf{0}. \label{eqMSE:Zb} 
			\end{align}
		\end{subequations}
		where $\Lambda_{C,k}$ are scalar weights and $\boldsymbol{\Lambda}_S$ is a matrix weight. 
		
		Next, we construct the Lagrangian functions for both problems. For problem \eqref{eq:MI}, the Lagrangian is given by
		\begin{equation}\label{eq:46}
			\begin{aligned}
				&\mathcal{L}_{(38)}\left(\boldsymbol{\alpha},\boldsymbol{W},\boldsymbol{\mu}\right)=-\kappa\sum_{k=1}^K\log_2\left(1+\frac{D_k}{N_k}\right)\\
				&-(1-\kappa)\log_2(1+\mathit{\Gamma})-\sum_{k=1}^Kf_\rho(\alpha_k)+\boldsymbol{\mu}^T\boldsymbol{g}(\boldsymbol{\alpha},\boldsymbol{W}),
			\end{aligned}
		\end{equation}
		with $\boldsymbol{\mu}\succeq 0$ denoting the Lagrange multipliers associated with the constraints.
		
		Similarly, the Lagrangian of the sum weighed MMSE reformulation \eqref{eq:MSE} is expressed as
		\begin{equation}\label{eq:47}
			\begin{aligned}
				\mathcal{L}_{(47)}(\boldsymbol{\alpha},\boldsymbol{W},\boldsymbol{\mu})=\;&\kappa\sum_{k=1}^K\Lambda_{C,k}e_{C,k}^\mathrm{MMSE}+(1-\kappa)\operatorname{tr}(\boldsymbol{\Lambda}_S\boldsymbol{E}_{S}^{\mathrm{MMSE}})\\
				&-\sum_{k=1}^Kf_\rho(\alpha_k)+\boldsymbol{\mu}^T\boldsymbol{g}(\boldsymbol{\alpha},\boldsymbol{W}).
			\end{aligned}
		\end{equation}

		By taking derivatives of $\mathcal{L}_{(38)}$ and $\mathcal{L}_{(47)}$ with respect to $\boldsymbol{W}$ and $\boldsymbol{\alpha}$, we obtain the stationarity conditions, which are given in \eqref{eq:LW1}–\eqref{eq:LA1} for the original formulation and \eqref{eq:LW2}–\eqref{eq:LA2} for the sum weighed MMSE formulation. Here, we define $\boldsymbol{M} \triangleq \boldsymbol{\Theta}\boldsymbol{G}$ and $\boldsymbol{B}\triangleq \boldsymbol{M}^H\boldsymbol{R}_H\boldsymbol{M}$. The derivatives $\frac{\partial(\boldsymbol{\mu}^{T}\mathbf{g}(\boldsymbol{\alpha},\boldsymbol{W}))}{\partial\boldsymbol{w}_{k}}$, $f_{\rho}^{\prime}(\alpha_k) \triangleq \frac{\partial f_{\rho}(\alpha_k)}{\partial \alpha_k}$ and $\frac{\partial(\boldsymbol{\mu}^T\boldsymbol{g}(\boldsymbol{\alpha},\boldsymbol{W}))}{\partial\alpha_k}$ are kept implicit since they do not affect the subsequent comparative analysis. A key observation is that, when the weight matrices are chosen as
		\begin{equation}\label{eq:48}
			\Lambda_{{C},k}=\left(e_{C,k}^\mathrm{MMSE}\right)^{-1},\;\boldsymbol{\Lambda}_S=\left(\boldsymbol{E}_{S}^{\mathrm{MMSE}}\right)^{-1},
		\end{equation}
		the gradients in \eqref{eq:LW1} and \eqref{eq:LW2}, as well as in \eqref{eq:LA1} and \eqref{eq:LA2}, coincide. Moreover, differentiating with respect to $\boldsymbol{\mu}$ yields identical feasibility conditions for both problems. This means that any point satisfying the KKT conditions of one problem also satisfies the KKT conditions of the other. Therefore, the two problems share the same set of first-order stationary points, establishing their equivalence from the perspective of first-order optimality. \hfill $\blacksquare$
	\end{proof}

	Thanks to Proposition 2, the original non-convex problem \eqref{eq:29} is transformed into a structured convex optimization problem which can be iteratively solved by alternately updating each variable block while fixing the others. In particular, the receiving combiners of both the communication and sensing sides directly follow the MMSE forms as derived in Proof of Proposition 2, with their explicit closed-form solutions provided in \eqref{eq:37} and \eqref{eq:38}, respectively.
	
	With the combiners fixed, the next step is to update the MSE weights. The communication-side MSE weights $\{\Lambda_{C,k}\}$ are updated as
	\begin{equation}\label{eq:49}
		\Lambda_{C,k}=\left(1-\frac{\alpha_k\boldsymbol{w}_k^H\boldsymbol{h}_{C,k}\boldsymbol{h}_{C,k}^H\boldsymbol{w}_k}{\sum_{j=1}^K\alpha_j\boldsymbol{h}_{C,k}^H\boldsymbol{w}_j\boldsymbol{w}_j^H\boldsymbol{h}_{C,k}+4\sigma_C^2}\right)^{-1}.
	\end{equation}
	while the sensing-side MSE weight matrix $\boldsymbol{\Lambda}_S$ is updated as 
	\begin{equation}\label{eq:50}
		\boldsymbol{\Lambda}_S=\left(\boldsymbol{R}_H\left(\boldsymbol{I}_{M}+\frac{1}{4N_R\sigma_S^2}\boldsymbol{V}_{S}\boldsymbol{s}\boldsymbol{s}^H\boldsymbol{V}_S^H\boldsymbol{R}_H\right)^{-1}\right)^{-1}.
	\end{equation}
	
		\begin{table*}[!b]
		\normalsize  
		\centering
		\hrule	\smallskip
		\begin{IEEEeqnarray}{rCl} \label{eq:51}
			\min_{\boldsymbol{\alpha},\boldsymbol{W}}\;&\kappa&\sum_{k=1}^K\Lambda_{C,k}\left(\left|\frac{1}{2}u_{C,k}\boldsymbol{h}_{C,k}^H\boldsymbol{w}_k-1\right|^2+\sum_{j\neq k}^K\left|\frac{1}{2}u_{C,k}\boldsymbol{h}_{{C},k}^H\boldsymbol{w}_j\right|^2\right)+(1-\kappa) \operatorname{tr}\left(\boldsymbol{\Lambda}_s\left(\bar{\boldsymbol{M}}^H\boldsymbol{R}_H\bar{\boldsymbol{M}}-\bar{\boldsymbol{M}}^H\boldsymbol{R}_H-\boldsymbol{R}_H\bar{\boldsymbol{M}}\right)\right)\notag\\
			&-&\sum_{k=1}^{K}f_{\rho}(\alpha_{k})\qquad \mathrm{s.t.}\;\eqref{eq27:Zb}, \eqref{eq27:Zc}, \eqref{eq27:Zd}.\IEEEeqnarraynumspace
		\end{IEEEeqnarray}
		\vspace{-0.5cm} 
	\end{table*}
	
	With the remaining variables fixed, the joint optimization of $\boldsymbol{\alpha}$ and $\boldsymbol{W}$ reduces to solving
	
	\begin{subequations}\label{eq:52}
		\begin{align}
			\min_{\boldsymbol{\alpha},\boldsymbol{W}}\;&\kappa\sum_{k=1}^K\Lambda_{C,k}e_{C,k}+(1-\kappa)\operatorname{tr}(\boldsymbol{\Lambda}_S\boldsymbol{E}_S)-\sum_{k=1}^Kf_\rho(\alpha_k)\label{eq37:Za}\\
			\mathrm{s.t.}\;&\eqref{eq27:Zb}, \eqref{eq27:Zc}, \eqref{eq27:Zd}. \label{eq37:Zb}
		\end{align}
	\end{subequations}
	
	Since the term $-f_\rho(\alpha_k)$ in the objective function is concave in $\alpha_k$, we linearize it at the current iteration via first-order Taylor expansion \cite{ref36}. Specifically, at the $p$-th iteration, we approximate
	\begin{equation}\label{eq:53}
		f_\rho(\alpha_k)\approx f_\rho\left(\alpha_k^{(p-1)}\right)+\left(\alpha_k-\alpha_k^{(p-1)}\right)\cdot\nabla f_\rho\left(\alpha_k^{(p-1)}\right).
	\end{equation}

	Substituting the expressions of $e_{C,k}$ and $\boldsymbol{E}_S$ into the objective function yields \eqref{eq:51}, which is shown at the bottom of the page, where we define $\bar{\boldsymbol{M}} \triangleq \tfrac{1}{2}\boldsymbol{\Theta}\boldsymbol{G}\boldsymbol{W}\boldsymbol{s}\boldsymbol{u}_S^H$ for notational brevity. Now, it turns to be a standard convex optimization problem that can be efficiently solved using CVX. 
	
	\begin{algorithm}[htbp]
		\caption{Joint User Scheduling and Air-Ground Omnidirectional S\&C Algorithm}\label{alg:2}
		\begin{algorithmic}[1]
			\State \textbf{Input:}$M$, $N_T$, $N_R$, $K$, $N_S$, $\boldsymbol{h}_{C,k}^H$, $\boldsymbol{R}_H$, $\boldsymbol{G}$, $P_T$, $\sigma_C^2$, $\sigma_S^2$, $i_{\max}$, $p_{\max}$, $\kappa$, $\rho$, $\epsilon$.
			\State \textbf{Set} outer iteration index $i = 1$, and \textbf{initialize} $\boldsymbol{\alpha}^{(0)}$, $\mathbf{W}^{(0)}$.
			\State \textbf{Repeat}
			\State Fix $\boldsymbol{\alpha}^{(i-1)}$ and $\boldsymbol{W}^{(i-1)}$, update $\boldsymbol{\Theta}^{(i)}$ using the RGA algorithm in \autoref{alg:1} or the SDR algorithm in \autoref{alg:3}.
			\State Set inner iteration index $p=1$.
			\State Fix $\boldsymbol{\Theta}^{(i)}$, update $\boldsymbol{\alpha}^{(i)}$ and $\boldsymbol{W}^{(i)}$:
			\State \hspace{1em} \textbf{Repeat}
			\State \hspace*{1em} (a) Update $u_{C,k}^{(p)}$, $\forall k$, via \eqref{eq:37}.\;	
			\State \hspace*{1em} (b) Update $\boldsymbol{u}_S^{(p)}$ via \eqref{eq:38}.\;
			\State \hspace*{1em} (c) Update $\Lambda_{C,k}^{(p)}$, $\forall k$, via \eqref{eq:49}.\;
			\State \hspace*{1em} (d) Update $\boldsymbol{\Lambda}_S^{(p)}$ via \eqref{eq:50}.\;
			\State \hspace*{0.8em} (e) Update $\boldsymbol{\alpha}^{(p)}$ and $\boldsymbol{W}^{(p)}$ by solving the convex problem \eqref{eq:51} using CVX.\;
			\State \hspace*{1em} (f) $p=p+1$.\;
			\State \hspace*{1em} \textbf{Until} convergence or $p \geq p_{\max}$\;
			\State $\boldsymbol{\alpha}^{(i)}=\boldsymbol{\alpha}^{(p)}$ and $\boldsymbol{W}^{(i)}=\boldsymbol{W}^{(p)}$.\;
			\State $i=i+1$.\;
			\State \textbf{Until} convergence or $i \geq i_{\max}$.
			\State \textbf{Output:} $\boldsymbol{\alpha}^\mathrm{opt}=\boldsymbol{\alpha}^{(i)}$,$\boldsymbol{W}^\mathrm{opt}=\boldsymbol{W}^{(i)}$ and $\boldsymbol{\Theta}^\mathrm{opt}=\boldsymbol{\Theta}^{(i)}$.
		\end{algorithmic}
	\end{algorithm}
	The entire iterative procedure, encompassing both the backward passive coefficient optimization and the joint scheduling and active beamforming design, is summarized in Algorithm \autoref{alg:2}, which we refer to as the joint user scheduling and air-ground omnidirectional (US-AGO) S\&C algorithm. Depending on the backward passive coefficient optimization algorithms employed in Subsection~\ref{subsec:SDR}, we label the algorithm as:
	\begin{itemize}
		\item US-AGO-R: using RGA algorithm, or
		\item US-AGO-S: using SDR algorithm.
	\end{itemize}

			\begin{figure}[t]
				\vspace{-0.4cm} 
				\centering
				\includegraphics[width=0.98\columnwidth]{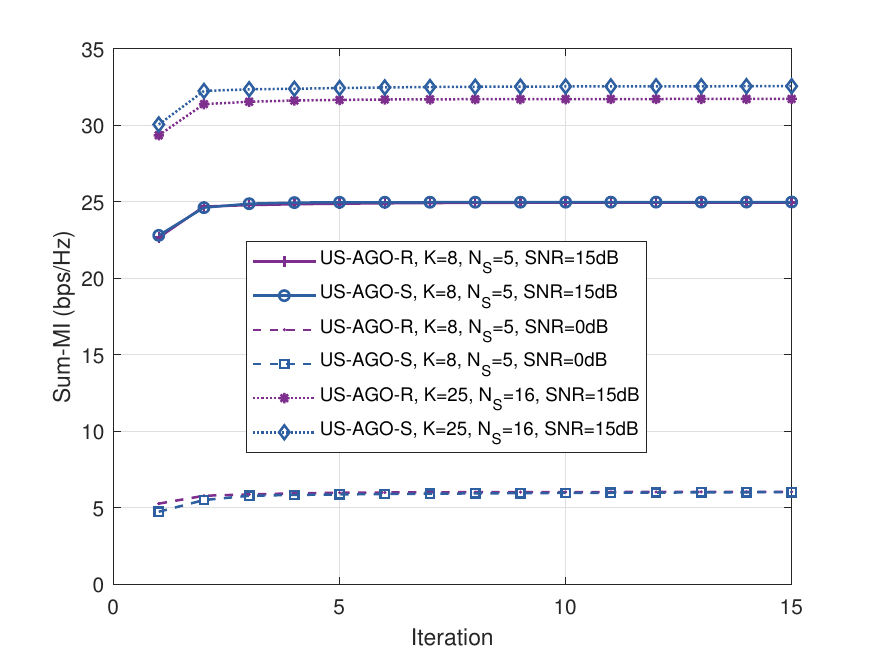} 
				\caption{
					The convergence behavior by the proposed US-AGO-R and US-AGO-S algorithms under various system configurations.
				}
				\label{fig3}
				\vspace{-0.2cm} 
			\end{figure}
			\begin{figure*}[!hb]
				\vspace{-0.4cm} 
				\centering
				\begin{minipage}[t]{0.49\textwidth}
					\centering
					\includegraphics[width=\linewidth]{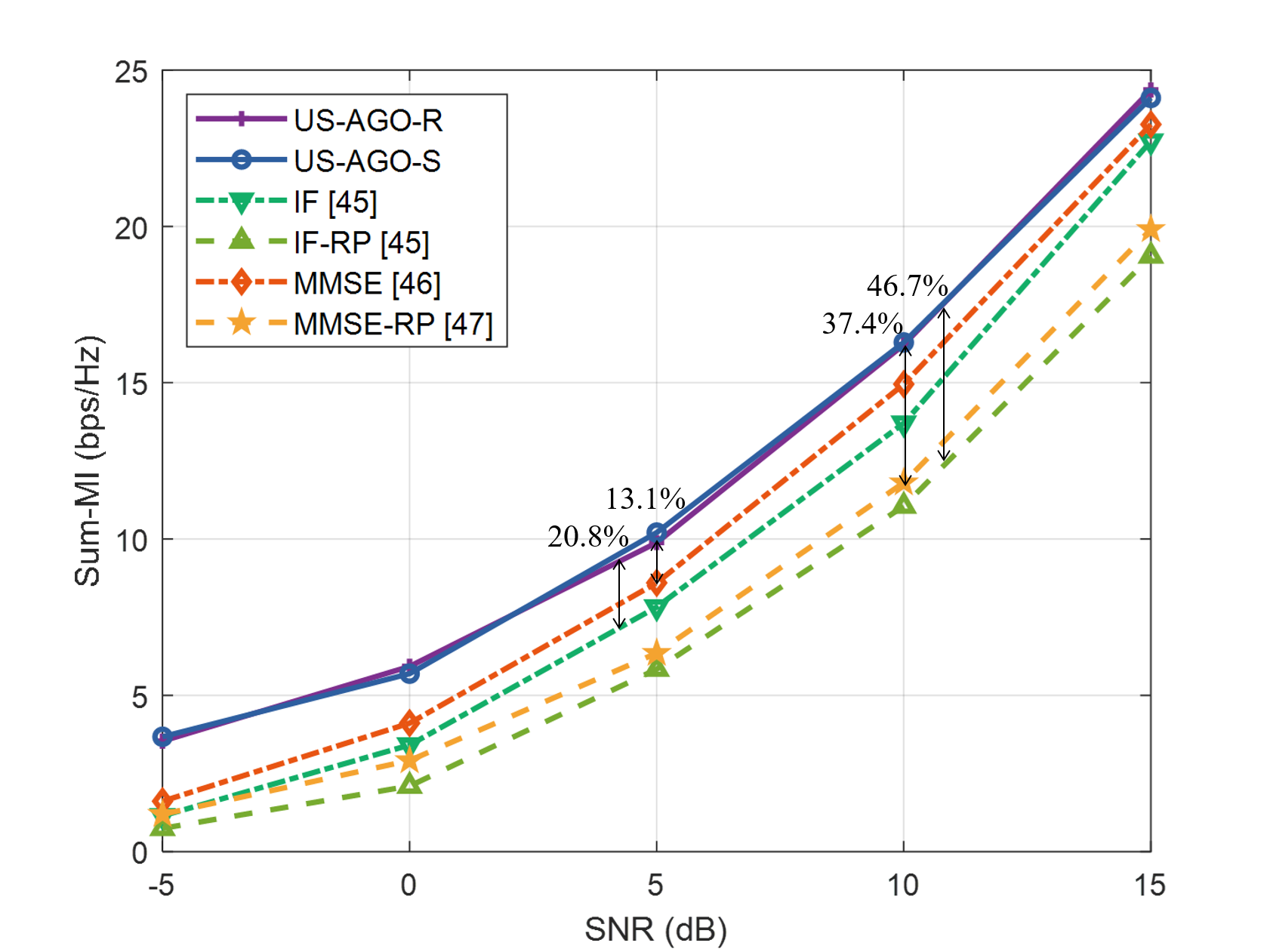}
					\caption{Sum-MI versus SNR by proposed algorithms and baselines.}
					\label{fig4}
				\end{minipage}
				\hfill
				\begin{minipage}[t]{0.49\textwidth}
					\centering
					\includegraphics[width=\linewidth]{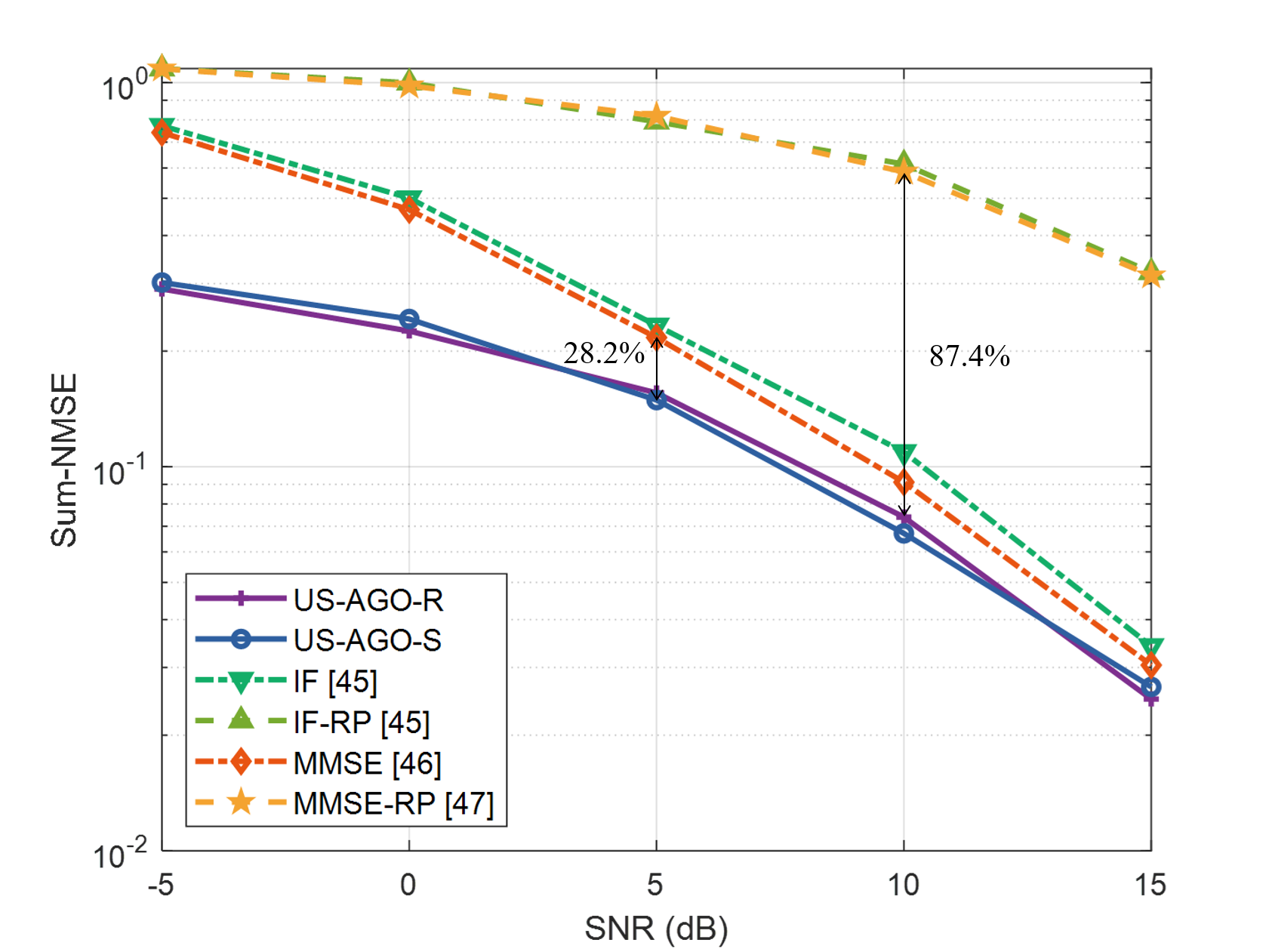}
					\caption{Sum-NMSE versus SNR by proposed algorithms and baselines.}
					\label{fig5}
				\end{minipage}
			\end{figure*}

	\subsection{Complexity Analysis}
	We now analyze the computational complexity of the proposed algorithms, which consists of an outer iterative procedure with two main stages: 
	(i) fixing $\boldsymbol{\alpha}$ and $\boldsymbol{W}$ to update $\boldsymbol{\Theta}$, and (ii) fixing $\boldsymbol{\Theta}$ to update $\boldsymbol{\alpha}$ and $\boldsymbol{W}$. 
			
	For the first stage, if the RGA algorithm is adopted, it updates $\boldsymbol{\Theta}$ on the complex circle manifold with a per-iteration cost of $\mathcal{O}(M^2)$, mainly due to matrix-vector products and projections \cite{ref33}. Assuming convergence within $q_{\max}$ iterations, the total cost is $\mathcal{O}(q_{\max}M^2)$.
	In contrast, the SDR algorithm requires solving an SDP with a Hermitian matrix $\boldsymbol{Q}\in\mathbb{C}^{M\times M}$, which has a worst-case complexity of $\mathcal{O}(M^{4.5}\mathrm{log}\left(1/\epsilon\right))$ under interior-point methods \cite{ref37}, where $\epsilon$ denotes the solution accuracy. The subsequent Gaussian randomization further contributes $\mathcal{O}(M^3+LM^2)$ operations, with $L$ being the number of random samples. Hence, the overall complexity is $\mathcal{O}(M^{4.5}\mathrm{log}\left(1/\epsilon\right)+LM^2)$, which is substantially higher than that of the RGA algorithm.
			
	In the second stage, updating $\boldsymbol{\alpha}$ and $\boldsymbol{W}$ involves an inner loop where the receive combiners $\{u_{C,k}\}$, $\boldsymbol{u}_{S}$ and MSE weights $\{\Lambda_{C,k}\}$, $\boldsymbol{\Lambda}_S$ are updated with respective complexities of $\mathcal{O}(K^2N_T)$, $\mathcal{O}(M^2+MK)$, $\mathcal{O}(K^2N_T)$ and $\mathcal{O}(M^3)$. The most computationally demanding step is solving the joint optimization problem \eqref{eq:39} via CVX, whose complexity scales as $\mathcal{O}((N_TK)^3)$. Given that  $K\ll M$, the overall complexity of this stage is dominated by $\mathcal{O}(j_{\max}((N_TK)^3+M^{3}))$.
			
	Therefore, the total complexity of US-AGO-R is $\mathcal{O}(i_{\max}(q_{\max}M^{2}+j_{\max}((N_{T}K)^{3}+M^{3})))$, whereas that of US-AGO-S is $\mathcal{O}(i_{\max}(M^{4.5}\log\left(\frac{1}{\epsilon}\right)+LM^2+j_{\max}((N_TK)^3+M^3)))$.

	\section{Numerical Results} \label{sec:results}
	In this section, we evaluate the performance of our proposed design. Unless otherwise specified, the default simulation setup considers an ISAC BS equipped with $M=N_T=N_R=16$ elements and $K=8$ CUs. The maximum number of simultaneously scheduled users is set to $N_S=5$. For the communication channel, we assume to have $N_{cl}=4$ clusters and $N_{ray}=10$ rays contained in each cluster. The transmit power is set to $P_T=40 \mathrm{dBm}$, and the carrier frequency is 3 GHz. The weighting factor $\kappa$, the penalty weight $\rho$, and the convergence threshold $\epsilon$ are set to 0.5, 0.1 and 0.01, respectively. The following benchmark algorithms are used:

	1) Interference free (IF) \cite{ref38}: In the joint scheduling and active beamforming stage, users are scheduled randomly, and the active beamformer is designed to eliminate inter-user interference by enforcing orthogonality in the effective channel matrix across scheduled users. Then, the backward passive coefficient matrix is optimized using the proposed algorithms.
			
	2) Interference free random phase (IF-RP) \cite{ref38}: Random phase shifts are applied for backward passive coefficients. The joint scheduling and active beamforming follow the same strategy as IF.
			
	3) MMSE \cite{ref39}: The backward passive coefficient matrix is optimized using the proposed algorithms, while user scheduling is performed randomly, and the active beamformer is designed to balance signal enhancement and interference suppression by minimizing the total MSE between transmitted and received signals.
			
	4) MMSE-RP \cite{ref40}: Random phase shifts are applied for backward passive coefficients. The joint scheduling and active beamforming strategy is the same as in MMSE.

	Fig.~\ref{fig3} illustrates the convergence performance of the proposed US-AGO-R and US-AGO-S algorithms under different SNR and user load conditions. It can be observed that both US-AGO-R and US-AGO-S converge rapidly within approximately 5 iterations in all considered scenarios, demonstrating their efficiency and suitability	for practical deployment.

		    \begin{figure*}[!hb]  
				\vspace{-0.4cm} 
				\centering
				\begin{subfigure}[t]{0.5\textwidth}
					\centering
					\includegraphics[width=\textwidth]{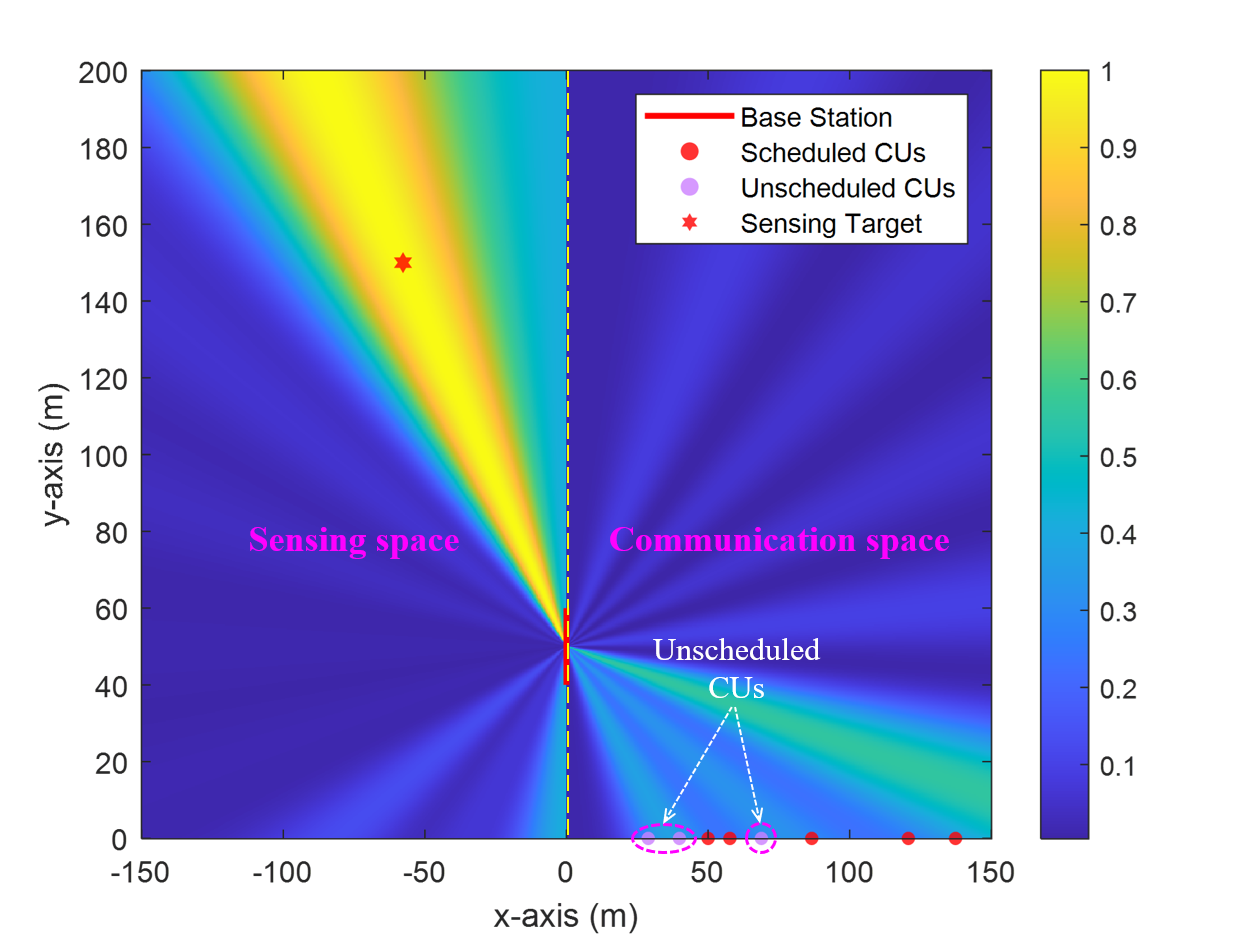}
					\caption{2D beampattern ($M=N_T=N_R=8$).}
					\label{fig:6a}
				\end{subfigure}
				\hfill
				\begin{subfigure}[t]{0.48\textwidth}
					\centering
					\includegraphics[width=\textwidth]{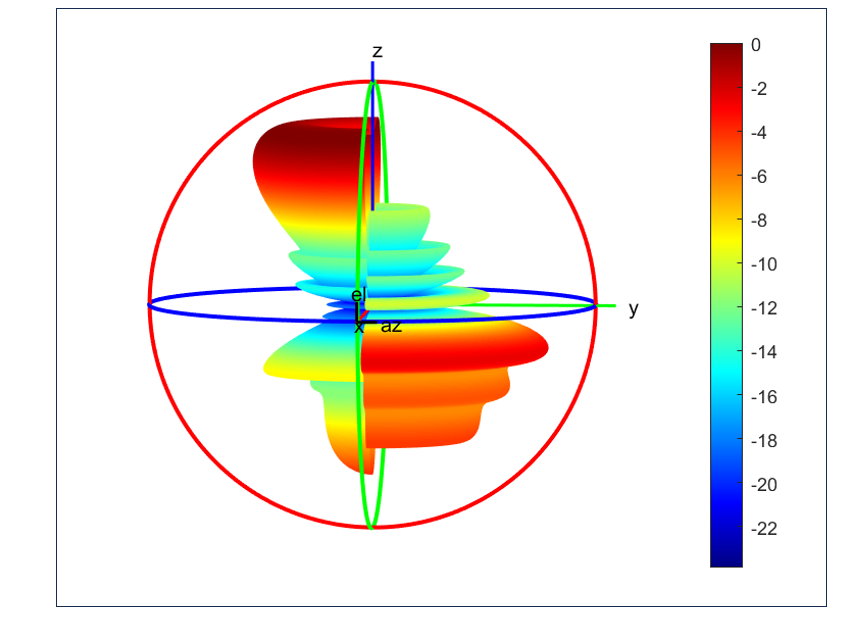}
					\caption{3D beampattern ($M=N_T=N_R=8$).}
					\label{fig:6b}
				\end{subfigure}
				\begin{subfigure}[t]{0.5\textwidth}
					\centering
					\includegraphics[width=\textwidth]{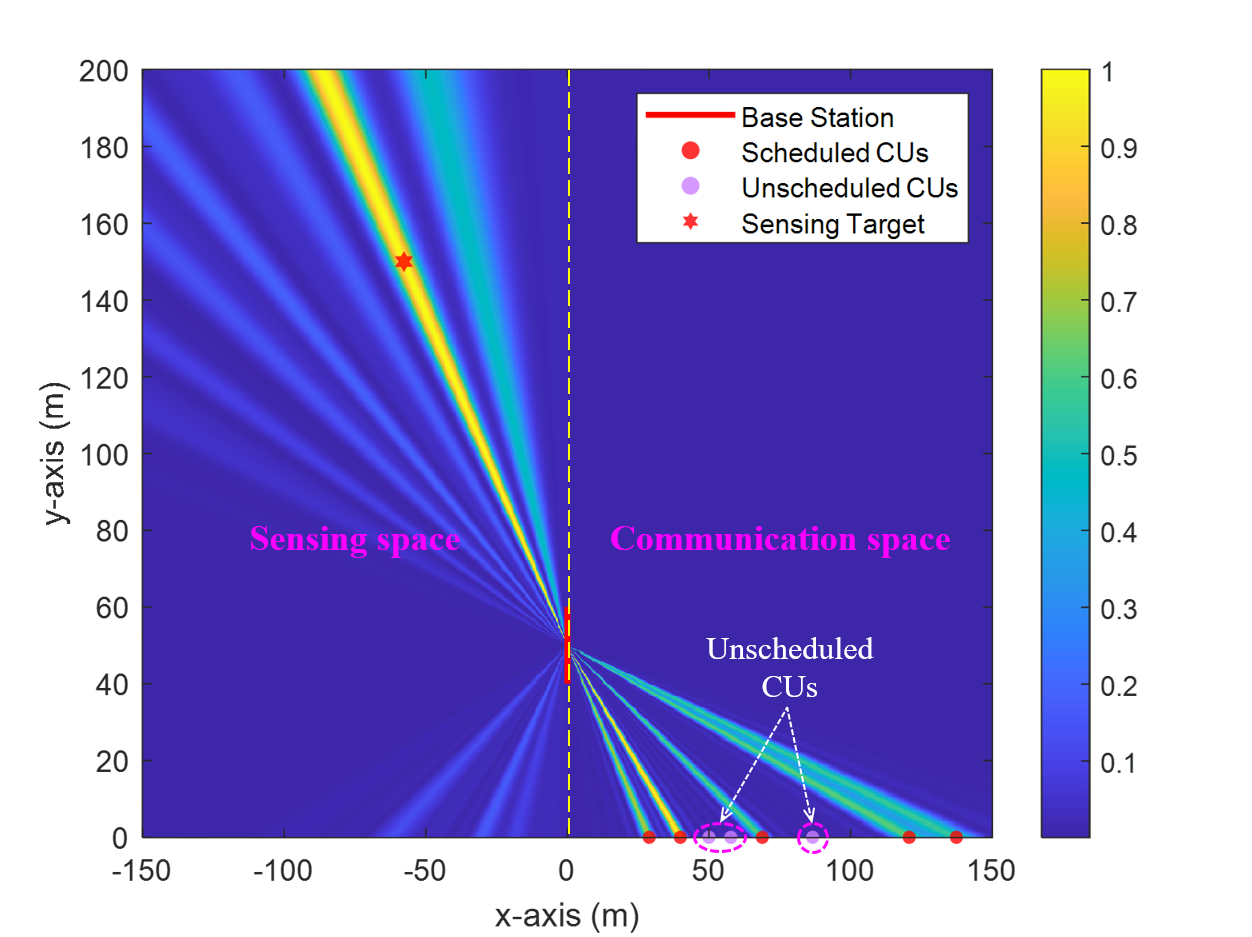}
					\caption{2D beampattern ($M=N_T=N_R=64$).}
					\label{fig:6c}
				\end{subfigure}
				\hfill
				\begin{subfigure}[t]{0.48\textwidth}
					\centering
					\includegraphics[width=\textwidth]{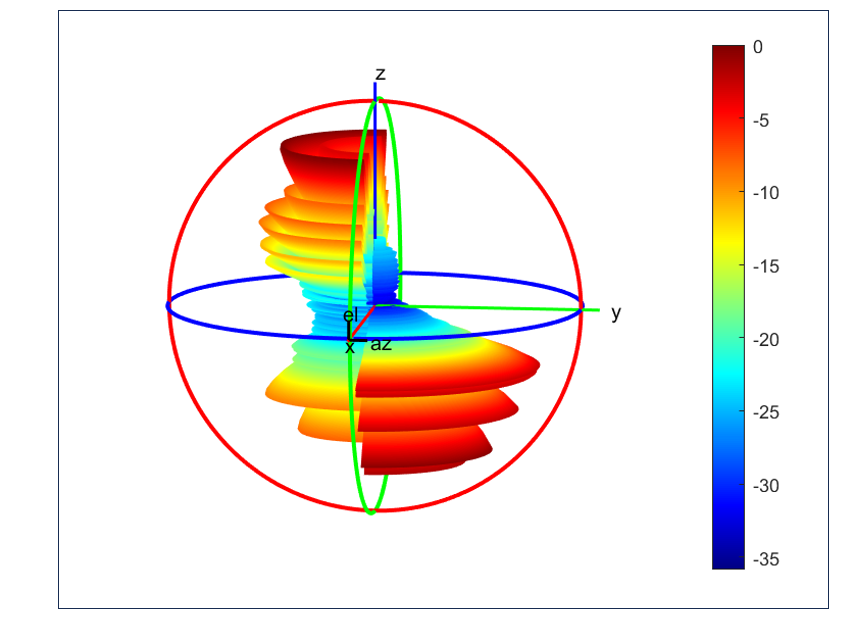}
					\caption{3D beampattern ($M=N_T=N_R=64$).}
					\label{fig:6d}
				\end{subfigure}
				\caption{Normalized 2D (x–y plane, linear scale) and 3D (azimuth/elevation, dB scale) beampatterns of the proposed air-ground OmniSteering antenna with $K=8$, $N_S=5$. (a)-(b) Results for $M=N_T=N_R=8$; (c)-(d) Results for $M=N_T=N_R=64$.}
				\label{fig6}
				\vspace{-0.4cm} 
			\end{figure*}
			
	Fig.~\ref{fig4} illustrates the system sum-MI performance under various SNR levels. The proposed US-AGO-R and US-AGO-S algorithms consistently achieve the best performance across all SNR conditions. By comparing the schemes with optimized backward passive coefficient matrix to their random-phase counterparts (i.e., -RP variants), a clear performance gap can be observed. This highlights the considerable benefit brought by backward passive coefficient optimization, further demonstrating its critical role in enhancing sensing performance. To quantify these improvements, we observe that at 5 dB, the proposed algorithms outperform MMSE and IF by 13.1\% and 20.8\%, respectively, demonstrating a notable advantage in the low-SNR regime. At 10 dB, the performance gap becomes even more significant, with the proposed algorithms surpassing MMSE-RP and IF-RP by 37.4\% and 46.7\%, respectively. These results validate the substantial performance gains achieved by the proposed design, particularly when compared to both conventional and random-phase benchmark schemes.

	In Fig.~\ref{fig5}, we evaluate the sum-NMSE performance of the proposed and baseline algorithms across different SNR levels. Before proceeding with the analysis, we define the sum-NMSE metric as the sum of the normalized mean square error in estimating the transmitted signal (communication task) and that in estimating the echoed channel (sensing task), given by
			\begin{equation}
				\mathrm{sum-NMSE}=\underbrace{\frac{\mathbb{E}[\|\hat{\boldsymbol{s}}-\boldsymbol{s}\|^2]}{\mathbb{E}[\|\boldsymbol{s}\|^2]}}_{\text{Communication NMSE}}+\underbrace{\frac{\mathbb{E}[\|\hat{\boldsymbol{H}}_S^H-\boldsymbol{H}_S^H\|^2]}{\mathbb{E}[\|\boldsymbol{H}_S^H\|^2]}}_{\text{Sensing NMSE}}.
			\end{equation}

	It shows that the proposed US-AGO-R and US-AGO-S algorithms consistently achieve the lowest sum-NMSE across the entire SNR regime. Meanwhile, the MMSE and IF algorithms exhibit similar performance to each other, as do their random-phase variants (i.e., MMSE-RP and IF-RP), though all with noticeably higher NMSE levels. The performance advantage of the proposed algorithms over MMSE and IF is particularly evident in the low-SNR regime, where they achieve the a sum-NMSE reduction of approximately 20.8\% at 5 dB. As the SNR increases, their advantage over the random-phase baselines (MMSE-RP and IF-RP) becomes more pronounced, reaching up to 87.4\% at 10 dB. These results further validate the effectiveness of the proposed algorithms in improving system performance under the sum-NMSE metric. And the consistency of conclusions drawn from both Fig.~\ref{fig4} and Fig.~\ref{fig5} provides additional evidence supporting the use of MI as a unified performance metric for ISAC.

	Fig.~\ref{fig6} depicts the normalized 2D and 3D beampatterns of the proposed air-ground OmniSteering antenna, clearly demonstrating its dual-functional beamforming capability. For both configurations, the BS is placed at (0, 50)~m, with eight CUs randomly distributed along the vertical line segment from (0, 0) to (0, 150)~m in the right half-plane (communication space). A single sensing target is located at $60^{\circ}$ elevation in the left half-plane (sensing space), where the omni-steering plate operates in T-mode.

	In Fig.~\ref{fig6}(a)-(b), a small-scale array with $M=N_T=N_R=8$ is employed. The 2D beampattern shown in Fig.~\ref{fig6}(a) reveals that the sensing beam is precisely directed toward the designated target, exhibiting a dominant mainlobe. Nevertheless, the limited aperture results in a relatively wide mainlobe, leading to poor angular resolution and reduced ability to distinguish closely spaced targets. In the communication space, our proposed algorithms successfully schedule five users according to the design requirements. There are three prominent beams pointing towards the ground users, which corresponds to the 3D beampattern in Fig.~\ref{fig6}(b).

	Fig.~\ref{fig6}(c)-(d) show the results for a large-scale configuration where $M=N_T=N_R=64$. The 2D beampattern in Fig.~\ref{fig6}(c) shows that five of eight CUs are successfully served with distinct and narrow beams, demonstrating both the effectiveness of the proposed scheduling algorithm and the substantially enhanced communication capability compared to Fig.~\ref{fig6}(a). In the sensing space, the mainlobe is more sharply focused on the target, achieving much higher angular resolution. The corresponding 3D beampattern in Fig.~\ref{fig6}(d) offers a more comprehensive view of the large-scale array’s beamforming performance, where the well-separated beams and suppressed sidelobes highlight the effectiveness of the design.

			\begin{figure}[t]
				\centering
				\includegraphics[width=0.98\columnwidth]{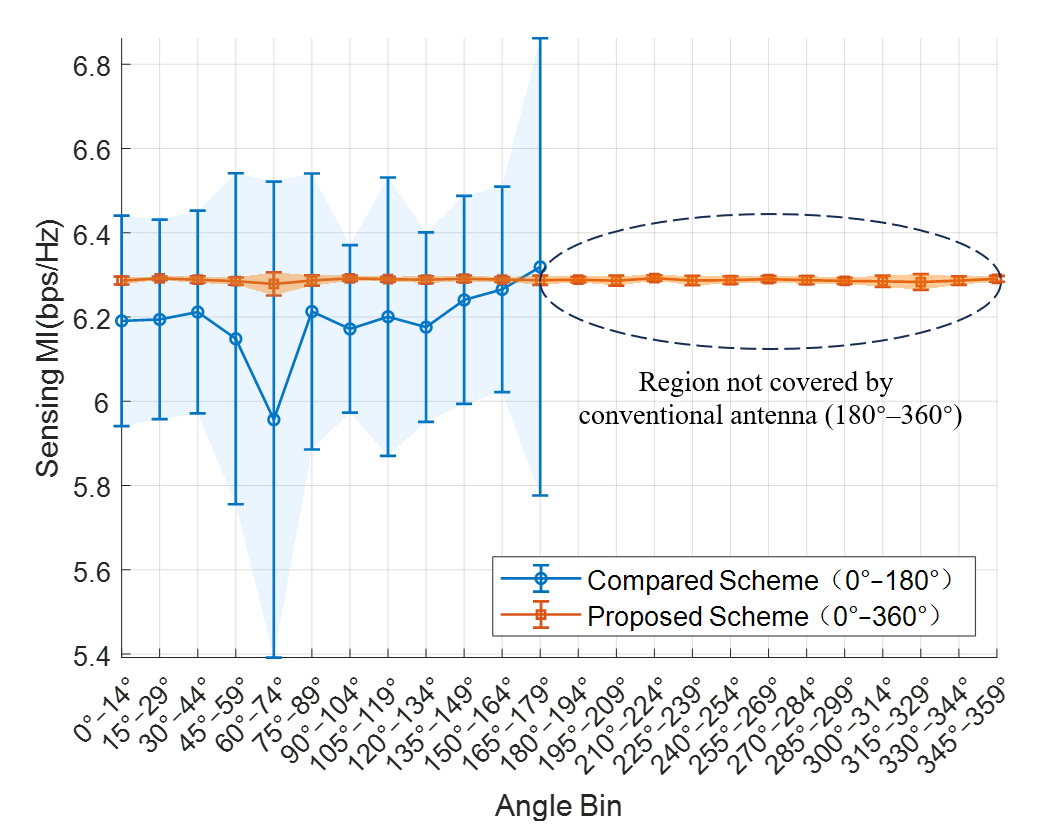} 
				\caption{
					Sensing MI versus angle bin.}
				\label{fig8}
				\vspace{-0.4cm} 
			\end{figure}
			
	Finally, to demonstrate the coverage advantage of the proposed antenna architecture, we evaluate the sensing MI performance across different angular bins, as shown in Fig.~\ref{fig8}\footnote{In this figure, each marker indicates the average sensing MI within the corresponding angular bin, while the associated bars represent the performance fluctuation, reflecting the stability of the sensing capability.}. A single-user scenario is considered to clearly highlight the sensing coverage capability. In the proposed scheme, the BS equipped with the air-ground OmniSteering antenna supports omnidirectional sensing over the entire $0^{\circ}-360^{\circ}$ angular space. In contrast, the conventional forward-facing ISAC antenna is inherently limited to front-side sensing coverage ($0^{\circ}-180^{\circ}$). Importantly, the proposed scheme maintains reliable sensing in all regions. It is attributed to the pre-determined power splitting strategy for backlobe and forwardlobe. These results clearly validate the proposed antenna’s ability to extend sensing coverage while ensuring robust and stable performance, effectively addressing the blind-zone limitation of conventional designs.

	\section{Conclusion} \label{sec:conclusion}
	In this work, a novel air-ground OmniSteering antenna structure is designed, where the conventional backlobe reflector is replaced by an omni-steering plate that adaptively modulates the backlobe, thereby overcoming coverage and beamforming limitations and enabling full-space services. Focusing on dense communication scenarios, we formulated a sum-MI maximization problem to jointly optimize user scheduling, active array beamforming, and passive coefficients of the omni-steering plate. The backward passive coefficient matrix is solved via low-complexity RGA or SDR, and the joint scheduling and active beamforming step is transformed into a tractable sum weighted MMSE problem for efficient iterative optimization. Numerical simulations show that our algorithm considerably outperforms the benchmarks, achieving up to 46.7\% improvement in sum-MI and 87.4\% reduction in sum-NMSE, while beampatterns confirm effective beam alignment and user scheduling performance. Moreover, the proposed antenna design is validated to enhance both coverage and sensing capability. The effectiveness of the proposed ISAC framework and highlights its potential for low-altitude economy applications.

		\end{document}